\documentclass[journal]{IEEEtran}

\usepackage{amsmath,amssymb,amsthm}
\usepackage{graphicx}
\usepackage{booktabs}
\usepackage{array}
\usepackage{multirow}
\usepackage{cite}
\usepackage{url}
\usepackage{tikz}

\newtheorem{lemma}{Lemma}
\newtheorem{corollary}{Corollary}

\usepackage[caption=false,font=footnotesize]{subfig}

\usetikzlibrary{arrows.meta,positioning}

\newtheorem{definition}{Definition}
\newtheorem{theorem}{Theorem}
\newtheorem{proposition}{Proposition}

\begin{document}

\title{Multiscale Visibility Theory:
A Task-Relative Operator-Geometric Framework for
Observation, Robustness, and Detectability}

\author{Tahir Cetin Akinci\\
\small University of California, Riverside, Bourns College of Engineering,
CE-CERT, Riverside, California, USA\\[1mm]
\small \text{Email: \text{tahircetin.akinci@ucr.edu}}}

\maketitle

\begin{abstract}
Multiscale Visibility Theory (MVT) is introduced as a task-relative
operator-geometric framework for determining how prescribed information
survives and remains accessible through a structured observation architecture.
Rather than evaluating observation quality by transform magnitude, output
energy, rank, or signal prominence alone, MVT measures the accessibility of a
specified task direction or task subspace through the row space of the
composite operator
$A_{\mathcal A}=R_EW_\rho H_\beta P_\Omega$.
Unlike projection-based diagnostic frameworks centered on discrimination
between nominal and fault-related system behavior, MVT tracks an independently
prescribed task as observation architecture, scale, representation, and
retained access vary. The formulation unifies exact stagewise retention,
task-subspace preservation, scale-indexed visibility and stability, worst-case
erasure robustness, collective visibility, covariance-normalized statistical
separation, and decision sufficiency while keeping geometric accessibility,
statistical detectability, and decision adequacy mathematically distinct.

A reproducible MATLAB implementation verifies the analytical relations using a
controlled synthetic benchmark with a weak 260-Hz prescribed task and a
stronger 380-Hz nuisance. Within the same frequency-indexed observation family,
task visibility is maximized at 260~Hz, whereas complete-signal observation
energy is maximized at 380~Hz. At the nominal noise level,
$D_{\max}=18>\Gamma_D=8.563847350668$, yielding a decision-visible interval of
258.5--261.5~Hz. A four-row phase-redundant observer preserves full task-model
visibility through two row erasures, while two individually subthreshold
branches at 258 and 262~Hz become jointly decision-visible with
$D_{12}=10.470833520729$. All 12 numerical audits satisfy the
$10^{-10}$ tolerance.
\end{abstract}


\section{Introduction}

Signal-processing methods describe signals through complementary mathematical constructions. Fourier analysis resolves global spectral content, short-time Fourier methods localize frequency content in time, wavelets provide multiresolution representations, scale-space theory organizes structures across resolution, and Hilbert--Huang analysis provides an adaptive instantaneous-frequency description for suitable nonstationary data \cite{cohen1995,daubechies1992,mallat2009,lindeberg1994,huang1998}. These methods characterize representation, localization, scale, energy, or instantaneous frequency. They do not, by themselves, answer whether a prescribed task direction remains accessible through a particular observation architecture or whether the retained evidence is sufficient for a required decision.

This distinction matters when a physically weak component is task-relevant while a stronger component is irrelevant to the task. A signal component can be present in the underlying process yet be suppressed by acquisition, attenuated or removed by a scale-dependent observation, discarded by a representation, lost after partial access, or retained geometrically but remain statistically insufficient. MVT is formulated to keep these layers separate rather than identifying task relevance with transform magnitude or output energy.

The central object in MVT is therefore not a new transform. A sensor array, filter bank, Fourier or wavelet front end, or another prescribed linear representation may define part of an effective observation family. MVT evaluates the task-relevant geometry induced by that family and then, only after an uncertainty model is introduced, connects the retained geometry to statistical and decision quantities.

The finite-dimensional development uses established tools from matrix analysis, principal-angle geometry, subspace perturbation, frame theory, covariance normalization, and statistical detection \cite{horn2013matrix,bjorckgolub1973,daviskahan1970,edelman1998,stewart1990,kay1998fundamentals,scharf1994,christensen2016,casazza2013,bodmann2005,holmes2004optimal,fickus2012numerically,fickus2014group}. Orthogonal-projection geometry has also been used explicitly in dynamic-system fault diagnosis. In particular, Ding, Li, and Liu formulate fault detection and isolation by projecting measurement data onto system subspaces in Hilbert space and by using projection residuals and gap-metric-based decision mechanisms \cite{DingLiLiu2026ProjectionFaultDiagnosis}. This provides an important precedent for projection-based system discrimination, but the problem addressed is different from that of MVT: projection-based fault diagnosis distinguishes nominal and fault-related system behavior, whereas MVT evaluates how a separately prescribed task direction or task subspace survives and remains accessible through a composed observation architecture. MVT further separates geometric accessibility from multiscale variation, partial-access robustness, statistical separation, and decision sufficiency.

Task-based acquisition and functional-observability studies provide additional related perspectives but address different design or dynamical-system questions \cite{shlezinger2019,bernardo2023,fernando2010,darouach2025}. MVT does not claim the underlying projection, principal-angle, perturbation, frame, or statistical-detection tools individually as new. Its contribution is their task-relative organization around a specified observation architecture and a sequence of distinct questions: structural survival, geometric accessibility, robustness, multiscale variation, collective access, statistical separation, and decision sufficiency.

The contributions of this paper are fivefold. First, a composite observation model defines single-task and task-subspace visibility without conflating visibility with gain, rank, conditioning, or output energy. Second, an exact conditional factorization attributes loss to acquisition, scale-dependent observation, representation, and retained access. Third, the same geometric construction is extended to scale-indexed subspace paths, worst-case row erasure, and collective observation. Fourth, covariance-normalized separation and an explicit decision threshold are connected to the geometric layer under stated noise and distributional assumptions. Fifth, the framework is tested by reproducible numerical audits and a controlled benchmark in which task visibility and signal-energy dominance are intentionally separated.

The term ``multiscale visibility'' has also appeared in application-specific literature with meanings unrelated to the operator-geometric definition used here \cite{chen2026}. The present formulation is likewise distinct from Shannon information measures \cite{shannon1948,cover2006}, classical functional observability \cite{fernando2010,darouach2025}, and active-subspace methods for identifying influential parameter directions \cite{constantine2014,constantine2017}. Table~\ref{tab:mvt_hierarchy} summarizes the layers that are kept separate throughout the paper.

\begin{table*}[t]
\caption{Hierarchy of the MVT framework and the question answered at each layer.}
\label{tab:mvt_hierarchy}
\centering
\scriptsize
\begin{tabular}{p{2.2cm}p{4.0cm}p{10.1cm}}
\toprule
Layer & Representative object & Question\\
\midrule
Structural & $As\neq0$, $\mathcal O_A=\mathcal R(A^H)$ & Does the prescribed direction survive the deterministic observation?\\
Geometric & $q_A(s)$, $q_{\min}(A;\mathcal S)$ & How much of the task, or of its least favorable model direction, lies in the accessible subspace?\\
Robust & $G_r(\mathcal S)$ and lower certificates & What task visibility survives admissible erasures?\\
Multiscale & $P_\beta$, $d_{\rm proj}$ & How does the accessible subspace move across scale and how much can task visibility change?\\
Collective & $V_b^-$, $\overline V_b$, $V_b^+$, $\delta q_{k|J}$ & What do complementary observations add and how sensitive is visibility to the retained branch pattern?\\
Statistical & $D_A$ & How strongly separated are task-conditioned observations after an uncertainty model is specified?\\
Decision & $D_A\ge\Gamma_D$ & Is the available statistical separation sufficient for the required operating point?\\
\bottomrule
\end{tabular}
\end{table*}

Table~\ref{tab:mvt_hierarchy} is an organizing structure rather than a collection of interchangeable quality indices. In particular, a favorable value at one layer does not automatically imply a favorable value at another.

\section{Observation Architecture and the Visibility Problem}

Let the ambient signal or parameter vector belong to $\mathbb F^n$, where $\mathbb F\in\{\mathbb R,\mathbb C\}$. An observation configuration is specified by
\begin{equation}
\mathcal A=(\Omega,\beta,\rho,E),
\label{eq:config}
\end{equation}
where $\Omega$ denotes acquisition availability, $\beta$ is a scale or observation parameter, $\rho$ indexes a representation, and $E$ denotes retained coordinates or channels. Equation~\eqref{eq:config} separates the architectural choices that will later be varied independently.

The corresponding effective linear observation operator is
\begin{equation}
A_{\mathcal A}=R_EW_\rho H_\beta P_\Omega .
\label{eq:composite}
\end{equation}
In \eqref{eq:composite}, $P_\Omega$ restricts acquisition, $H_\beta$ applies the scale-dependent observation, $W_\rho$ provides the selected representation, and $R_E$ models partial access after representation. Figure~\ref{fig:mvt_concept} shows this sequence and its relation to the final decision layer.

The observation-accessible subspace is

\begin{equation}
\mathcal O_{\mathcal A}=\mathcal R(A_{\mathcal A}^{H}),
\label{eq:accessible}
\end{equation}
with orthogonal projector
\begin{equation}
P_{\mathcal A}=P_{\mathcal O_{\mathcal A}}.
\label{eq:projector}
\end{equation}

Equations~\eqref{eq:accessible} and \eqref{eq:projector} identify the ambient directions that can influence the deterministic observation and provide the geometric object used by all subsequent visibility measures.

If a component is annihilated upstream, deterministic downstream processing cannot restore it because
\begin{equation}
\mathcal N(A)\subseteq\mathcal N(BA)
\label{eq:null_inclusion}
\end{equation}
for every compatible deterministic linear operator $B$. Equation~\eqref{eq:null_inclusion} is the structural reason redundancy or representation cannot create task information that was already lost before that stage.

As summarized in Fig.~\ref{fig:mvt_concept}, MVT does not replace the
underlying sensing or signal representation. Instead, it introduces a
task-relative analysis layer that determines what prescribed information
remains accessible through that representation, whether that accessibility
is robust to scale variation or partial loss, and whether the retained
information is sufficient for the required decision.

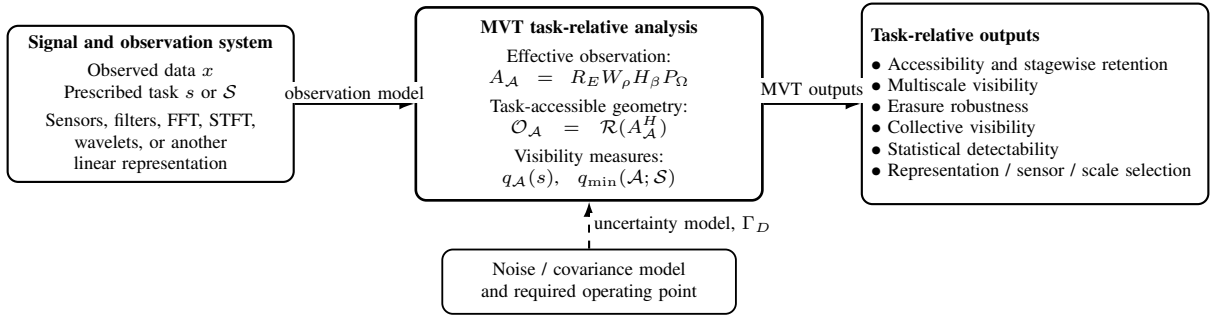
\begin{figure*}[t]
\centering

\begin{tikzpicture}[
>=Latex,
every node/.style={font=\scriptsize},
arr/.style={-{Latex[length=2mm]},thick},
darr/.style={-{Latex[length=2mm]},thick,dashed}
]

\node[
draw,
rounded corners,
line width=0.8pt,
align=center,
text width=3.55cm,
minimum height=2.0cm
] (obs) at (0,0)
{
\textbf{Signal and observation system}\\[1mm]
Observed data $x$\\
Prescribed task $s$ or $\mathcal{S}$\\[1mm]
Sensors, filters, FFT, STFT,\\
wavelets, or another linear representation
};

\node[
draw,
rounded corners,
line width=1.0pt,
align=center,
text width=4.35cm,
minimum height=2.55cm
] (mvt) at (5.8,0)
{
\textbf{MVT task-relative analysis}\\[1mm]
Effective observation:\\
$A_{\mathcal A}=R_EW_\rho H_\beta P_\Omega$\\[1mm]
Task-accessible geometry:\\
$\mathcal O_{\mathcal A}=\mathcal R(A_{\mathcal A}^{H})$\\[1mm]
Visibility measures:\\
$q_{\mathcal A}(s)$,\quad
$q_{\min}(\mathcal A;\mathcal S)$
};

\node[
draw,
rounded corners,
line width=0.8pt,
align=left,
text width=4.35cm,
minimum height=2.65cm
] (out) at (11.7,0)
{
\textbf{Task-relative outputs}\\[1mm]
$\bullet$ Accessibility and stagewise retention\\
$\bullet$ Multiscale visibility\\
$\bullet$ Erasure robustness\\
$\bullet$ Collective visibility\\
$\bullet$ Statistical detectability\\
$\bullet$ Representation / sensor / scale selection
};

\node[
draw,
rounded corners,
line width=0.7pt,
align=center,
text width=3.65cm,
minimum height=0.85cm
] (noise) at (5.8,-2.35)
{
Noise / covariance model\\
and required operating point
};


\draw[arr]
(obs.east) --
node[
above,
font=\scriptsize,
fill=white,
inner sep=1.5pt
]
{observation model}
(mvt.west);

\draw[arr]
(mvt.east) --
node[
above,
font=\scriptsize,
fill=white,
inner sep=1.5pt
]
{MVT outputs}
(out.west);

\draw[darr]
(noise.north) --
node[
right,
font=\scriptsize,
fill=white,
inner sep=1.5pt
]
{uncertainty model, $\Gamma_D$}
(mvt.south);

\end{tikzpicture}

\caption{
Conceptual role of Multiscale Visibility Theory (MVT).
Existing sensing and signal representations define or contribute to the
observation model supplied to MVT rather than competing with MVT as
estimators of the same quantity.
For a prescribed task direction $s$ or task subspace $\mathcal S$, MVT
constructs the task-accessible geometry
$\mathcal R(A_{\mathcal A}^{H})$ and evaluates task-relative visibility.
The resulting quantities characterize where task-relevant information is
retained or lost, how accessibility changes across scale, how robust it is
to partial loss, and what complementary observations add.
After an uncertainty model and required operating point are specified,
geometric accessibility is connected to statistical separation
$D_{\mathcal A}$ and decision sufficiency
$D_{\mathcal A}\ge\Gamma_D$.
The same quantities can support task-relative selection of representation,
scale, sensor configuration, or retained-access pattern.
}
\label{fig:mvt_concept}

\end{figure*}

\section{Geometric Visibility and Task-Model Preservation}
\label{sec:geometric_visibility}

The geometric layer of MVT quantifies whether a prescribed task is accessible through a specified observation architecture. The analysis is performed in the ambient signal space and depends only on the observation-accessible subspace and the prescribed task direction or task subspace. No statistical noise model is required at this stage.

\subsection{Single-Task Visibility}

Let $s\in\mathbb F^n$, $s\neq 0$, denote a prescribed task direction, where $\mathbb F\in\{\mathbb R,\mathbb C\}$. For the observation configuration $\mathcal A$, let
\begin{equation}
\mathcal O_{\mathcal A}
=
\mathcal R(A_{\mathcal A}^{H}),
\qquad
P_{\mathcal A}
=
P_{\mathcal O_{\mathcal A}}
\end{equation}
denote the observation-accessible subspace and its orthogonal projector, respectively.

\begin{definition}[Task-Relative Geometric Visibility]
The geometric visibility of the prescribed task direction $s$ under observation configuration $\mathcal A$ is
\begin{equation}
q_{\mathcal A}(s)
=
\frac{\|P_{\mathcal A}s\|^2}
{\|s\|^2}.
\label{eq:q}
\end{equation}
\end{definition}

Equation~\eqref{eq:q} measures the fraction of the squared task norm that lies in the observation-accessible subspace. It is therefore a task-relative geometric quantity rather than a measure of signal amplitude, output energy, or operator gain.

Because $P_{\mathcal A}$ is an orthogonal projector,
\begin{equation}
0
\le
q_{\mathcal A}(s)
\le
1.
\label{eq:q_bounds}
\end{equation}
The endpoints in \eqref{eq:q_bounds} have direct geometric interpretations. The condition $q_{\mathcal A}(s)=0$ means that $s$ is orthogonal to $\mathcal O_{\mathcal A}$ and is therefore completely inaccessible through the observation. The condition $q_{\mathcal A}(s)=1$ means that $s\in\mathcal O_{\mathcal A}$ and is geometrically preserved.

Let $\theta(s,\mathcal O_{\mathcal A})\in[0,\pi/2]$ denote the angle between the vector $s$ and the accessible subspace $\mathcal O_{\mathcal A}$. Then
\begin{equation}
q_{\mathcal A}(s)
=
\cos^2
\theta(s,\mathcal O_{\mathcal A}).
\label{eq:angle}
\end{equation}
Equation~\eqref{eq:angle} shows that geometric visibility is determined by subspace alignment. In particular, multiplying the observation operator by any nonzero scalar changes the observation magnitude but leaves $\mathcal R(A_{\mathcal A}^{H})$ unchanged; hence $q_{\mathcal A}(s)$ is invariant to such a global gain change.

This distinction is central to MVT. A task component may produce a small observation magnitude while remaining geometrically accessible, or it may produce substantial output energy through other signal components while the prescribed task direction itself has low visibility.

\subsection{Task-Subspace Visibility}

Many applications require preservation of a set of admissible task directions rather than a single vector. Let $\mathcal S\subseteq\mathbb F^n$ be a $d$-dimensional task subspace with orthonormal basis
\begin{equation}
U_S
=
\begin{bmatrix}
u_1 & \cdots & u_d
\end{bmatrix},
\qquad
U_S^HU_S=I_d.
\label{eq:US}
\end{equation}
Equation~\eqref{eq:US} provides coordinates for restricting the observation geometry to the prescribed task model.

\begin{definition}[Task-Subspace Visibility Matrix]
The task-subspace visibility matrix associated with $\mathcal A$ and $\mathcal S$ is
\begin{equation}
G_{\mathcal A,\mathcal S}
=
U_S^H
P_{\mathcal A}
U_S.
\label{eq:G}
\end{equation}
\end{definition}

Equation~\eqref{eq:G} is Hermitian positive semidefinite, and all of its eigenvalues lie in $[0,1]$. These eigenvalues quantify the geometric accessibility of orthogonal task-model directions after restriction to $\mathcal S$.

\begin{definition}[Worst-Direction Task Visibility]
The worst-direction visibility of the task subspace $\mathcal S$ is
\begin{equation}
q_{\min}(\mathcal A;\mathcal S)
=
\lambda_{\min}
\left(
G_{\mathcal A,\mathcal S}
\right).
\label{eq:qmin}
\end{equation}
\end{definition}

Equation~\eqref{eq:qmin} measures the least accessible unit-norm direction in the prescribed task model. Equivalently,
\begin{equation}
q_{\min}(\mathcal A;\mathcal S)
=
\min_{\substack{s\in\mathcal S\\ \|s\|=1}}
q_{\mathcal A}(s).
\label{eq:qmin_variational}
\end{equation}
The variational form in \eqref{eq:qmin_variational} makes clear that $q_{\min}$ is a worst-case task-preservation criterion rather than an average visibility measure.

The eigenvalues of $G_{\mathcal A,\mathcal S}$ are related to the principal angles between the task subspace and the accessible subspace \cite{bjorckgolub1973}. If
\begin{equation}
\theta_1
\le
\theta_2
\le
\cdots
\le
\theta_d
\end{equation}
denote the principal angles between $\mathcal S$ and $\mathcal O_{\mathcal A}$ when the required dimensions are available, then
\begin{equation}
\lambda_i
\left(
G_{\mathcal A,\mathcal S}
\right)
=
\cos^2\theta_i.
\label{eq:principal_eigenvalues}
\end{equation}
Consequently,
\begin{equation}
q_{\min}(\mathcal A;\mathcal S)
=
\cos^2\theta_{\max},
\label{eq:qmin_angle}
\end{equation}
where $\theta_{\max}$ is the largest principal angle associated with the task subspace. Equation~\eqref{eq:qmin_angle} identifies the least aligned task direction as the limiting direction for geometric task preservation.

\subsection{Survival, Restricted Injectivity, and Exact Preservation}

Task accessibility can be stated at several levels of strength. For a specific prescribed direction $s$, deterministic survival requires
\begin{equation}
A_{\mathcal A}s
\neq
0.
\label{eq:direction_survival}
\end{equation}
Equation~\eqref{eq:direction_survival} only states that the selected direction is not annihilated by the observation operator.

For a task subspace $\mathcal S$, restricted injectivity requires
\begin{equation}
\mathcal S
\cap
\mathcal N(A_{\mathcal A})
=
\{0\}.
\label{eq:restricted_injective}
\end{equation}
Equation~\eqref{eq:restricted_injective} ensures that no nonzero task-model direction is mapped to zero, although the task subspace need not be exactly contained in the accessible row space.

Exact geometric preservation is stronger and requires
\begin{equation}
\mathcal S
\subseteq
\mathcal R(A_{\mathcal A}^{H}).
\label{eq:exact_preserve}
\end{equation}
Under \eqref{eq:exact_preserve}, every direction in the task model is contained in the observation-accessible subspace, and therefore
\begin{equation}
q_{\min}(\mathcal A;\mathcal S)
=
1.
\label{eq:exact_qmin}
\end{equation}
Equation~\eqref{eq:exact_qmin} gives the MVT criterion for complete geometric preservation of the prescribed task subspace.

The three conditions in \eqref{eq:direction_survival}--\eqref{eq:exact_preserve} answer different questions. Direction survival concerns one specified task vector, restricted injectivity concerns uniqueness over the entire task model, and exact preservation requires full geometric inclusion. None of these conditions requires the observation operator to be invertible over the complete ambient space when only a lower-dimensional task subspace is relevant.

The principal quantities introduced in this section and used throughout the paper are summarized in Table~\ref{tab:core}.

\begin{table}[t]
\caption{Core MVT Quantities and Their Roles}
\label{tab:core}
\centering
\scriptsize
\begin{tabular}{p{1.35cm}p{3.05cm}p{2.35cm}}
\toprule
Quantity & Definition & Interpretation\\
\midrule

$q_{\mathcal A}(s)$
&
$\|P_{\mathcal A}s\|^2/\|s\|^2$
&
Single-task geometric visibility
\\

$q_{\min}$
&
$\lambda_{\min}(U_S^HP_{\mathcal A}U_S)$
&
Worst-direction task visibility
\\

$G_r(\mathcal S)$
&
Worst retained-set $q_{\min}$ under $\le r$ erasures
&
Task-relative erasure robustness
\\

$D_{\mathcal A}$
&
Covariance-normalized task separation
&
Statistical visibility
\\

$\Gamma_D$
&
Required statistical separation
&
Decision threshold
\\

$d_{\rm proj}$
&
$\|P_1-P_2\|_2$
&
Accessible-subspace variation
\\

\bottomrule
\end{tabular}
\end{table}

Table~\ref{tab:core} emphasizes that the geometric quantities introduced in this section are not interchangeable with robustness, statistical separation, or decision sufficiency. The subsequent sections connect these layers under additional assumptions while preserving their distinct mathematical roles.

\section{Exact Stagewise Visibility Factorization}
\label{sec:stagewise_factorization}

A final visibility value quantifies how much of a prescribed task remains accessible after the complete observation chain, but it does not identify the stage at which accessibility was reduced. MVT therefore introduces a conditional factorization in which each stage is evaluated relative to the task information that survived all preceding stages. This avoids attributing the same upstream loss repeatedly to downstream operations.

For the composite observation architecture in \eqref{eq:composite}, define the cumulative stage operators
\begin{align}
A_0 &= I, \nonumber\\
A_1 &= P_\Omega, \nonumber\\
A_2 &= H_\beta P_\Omega, \nonumber\\
A_3 &= W_\rho H_\beta P_\Omega, \nonumber\\
A_4 &= R_EW_\rho H_\beta P_\Omega .
\label{eq:cumulative_stage_operators}
\end{align}
Equation~\eqref{eq:cumulative_stage_operators} represents the observation chain after successively introducing acquisition, scale-dependent observation, representation, and retained access. Since each downstream operator acts on the output of the preceding stage,
\begin{equation}
\mathcal R(A_4^H)
\subseteq
\mathcal R(A_3^H)
\subseteq
\mathcal R(A_2^H)
\subseteq
\mathcal R(A_1^H)
\subseteq
\mathbb F^n .
\label{eq:stage_nested_subspaces}
\end{equation}
The nested relation in \eqref{eq:stage_nested_subspaces} shows that deterministic downstream processing cannot enlarge the task-accessible row space beyond what remains available from the preceding stage.

Let
\begin{equation}
P_j
=
P_{\mathcal R(A_j^H)},
\qquad
j=0,\ldots,4,
\label{eq:stage_projectors}
\end{equation}
with $P_0=I$. The retained task energy after stage $j$ is defined geometrically as
\begin{equation}
V_j
=
\|P_j s\|^2,
\qquad
j=0,\ldots,4.
\label{eq:stage_energy}
\end{equation}
Equation~\eqref{eq:stage_energy} is a squared projection norm in the ambient task space; it should not be confused with physical signal energy at the output of the corresponding processing block. From \eqref{eq:stage_nested_subspaces},
\begin{equation}
V_0
\ge
V_1
\ge
V_2
\ge
V_3
\ge
V_4
\ge
0.
\label{eq:stage_monotonicity}
\end{equation}
Thus, \eqref{eq:stage_monotonicity} expresses cumulative loss of geometric task accessibility along the deterministic observation chain.

The cumulative visibility after stage $j$ is
\begin{equation}
q_j(s)
=
\frac{V_j}{V_0}
=
\frac{\|P_j s\|^2}{\|s\|^2}.
\label{eq:cumulative_stage_visibility}
\end{equation}
Equation~\eqref{eq:cumulative_stage_visibility} extends the single-operator definition in \eqref{eq:q} to each intermediate stage of the observation architecture.

For stages whose upstream retained energy is nonzero, define the conditional retention factors
\begin{align}
\eta_\Omega
&=
\frac{V_1}{V_0},
&
\kappa_{\beta|\Omega}
&=
\frac{V_2}{V_1},
\nonumber\\[1mm]
\xi_{\rho|\beta,\Omega}
&=
\frac{V_3}{V_2},
&
q_{E|\rho,\beta,\Omega}
&=
\frac{V_4}{V_3}.
\label{eq:stage_factors}
\end{align}
Equation~\eqref{eq:stage_factors} attributes each loss only to the incremental stage under consideration. The factors therefore answer different questions: $\eta_\Omega$ measures acquisition retention, $\kappa_{\beta|\Omega}$ measures additional scale-dependent retention after acquisition, $\xi_{\rho|\beta,\Omega}$ measures additional representation-stage retention, and $q_{E|\rho,\beta,\Omega}$ measures the effect of retained access after all preceding stages.

Because of \eqref{eq:stage_monotonicity}, every defined conditional factor satisfies
\begin{equation}
0
\le
\eta_\Omega,\,
\kappa_{\beta|\Omega},\,
\xi_{\rho|\beta,\Omega},\,
q_{E|\rho,\beta,\Omega}
\le
1.
\label{eq:stage_factor_bounds}
\end{equation}
The bounds in \eqref{eq:stage_factor_bounds} provide a direct interpretation: a value of one indicates that the stage introduces no additional geometric loss for the surviving task component, whereas smaller values indicate incremental loss at that stage.

\begin{proposition}[Exact Conditional Stagewise Factorization]
\label{prop:stagewise_factorization}
Assume that $V_0,V_1,V_2,V_3>0$. The final task-relative visibility of the complete observation architecture satisfies
\begin{equation}
q_{E,\rho,\beta,\Omega}(s)
=
\eta_\Omega
\kappa_{\beta|\Omega}
\xi_{\rho|\beta,\Omega}
q_{E|\rho,\beta,\Omega}.
\label{eq:factor}
\end{equation}
\end{proposition}

\begin{proof}
Using the conditional definitions in \eqref{eq:stage_factors},
\begin{align}
\eta_\Omega
\kappa_{\beta|\Omega}
\xi_{\rho|\beta,\Omega}
q_{E|\rho,\beta,\Omega}
&=
\frac{V_1}{V_0}
\frac{V_2}{V_1}
\frac{V_3}{V_2}
\frac{V_4}{V_3}
\nonumber\\
&=
\frac{V_4}{V_0}.
\label{eq:stage_telescoping}
\end{align}
From \eqref{eq:cumulative_stage_visibility}, $V_4/V_0$ is precisely the geometric visibility of the complete operator
$A_4=R_EW_\rho H_\beta P_\Omega$. Hence \eqref{eq:stage_telescoping} gives \eqref{eq:factor}.
\end{proof}

Proposition~\ref{prop:stagewise_factorization} is an exact telescoping identity rather than an approximation or an independence assumption. If some upstream value $V_j$ becomes zero, the nesting in \eqref{eq:stage_nested_subspaces} implies $V_k=0$ for every $k>j$. The final visibility is then zero, and downstream ratios with zero denominators are left undefined rather than assigned artificial values.

The controlled numerical verification is summarized in Table~\ref{tab:stage}. The acquisition stage retains approximately $94.02\%$ of the original task geometry, and the scale-dependent stage retains approximately $87.97\%$ of what remains after acquisition. In this example, the representation stage has conditional factor $\xi_{\rho|\beta,\Omega}=1$, so it introduces no additional geometric loss. The retained-access stage has the smallest conditional factor, $q_{E|\rho,\beta,\Omega}=0.470272$, and therefore accounts for the largest incremental reduction in this specific observation chain.

\begin{table}[t]
\caption{Controlled Stagewise Visibility Verification}
\label{tab:stage}
\centering
\scriptsize
\begin{tabular}{lccc}
\toprule
Stage & Retained task energy & Cumulative & Conditional\\
\midrule
Full task
& 1.000000
& 1.000000
& --\\
Acquisition
& 0.940187
& 0.940187
& 0.940187\\
Scale observation
& 0.827103
& 0.827103
& 0.879722\\
Representation
& 0.827103
& 0.827103
& 1.000000\\
Retained access
& 0.388963
& 0.388963
& 0.470272\\
\bottomrule
\end{tabular}
\end{table}

As reported in Table~\ref{tab:stage}, the final cumulative visibility is
\begin{equation}
q_{E,\rho,\beta,\Omega}(s)
=
0.388962986023.
\label{eq:stage_direct_result}
\end{equation}
Equation~\eqref{eq:stage_direct_result} is obtained directly from the final accessible subspace.

Independently, the product of the four conditional factors gives
\begin{equation}
\eta_\Omega
\kappa_{\beta|\Omega}
\xi_{\rho|\beta,\Omega}
q_{E|\rho,\beta,\Omega}
=
0.388962986023.
\label{eq:stage_product_result}
\end{equation}
The agreement between \eqref{eq:stage_direct_result} and \eqref{eq:stage_product_result} has a numerical residual of $0.000\times10^{0}$ at the reported precision, confirming the implementation of Proposition~\ref{prop:stagewise_factorization} for the controlled case.

Figure~\ref{fig:stagewise} provides complementary cumulative and conditional views of the same factorization. The cumulative panel shows the progressive reduction of task accessibility through the complete chain, whereas the conditional panel identifies the incremental contribution of each stage without recounting losses that occurred earlier.

\begin{figure*}[t]
\centering
\includegraphics[width=0.84\textwidth]
{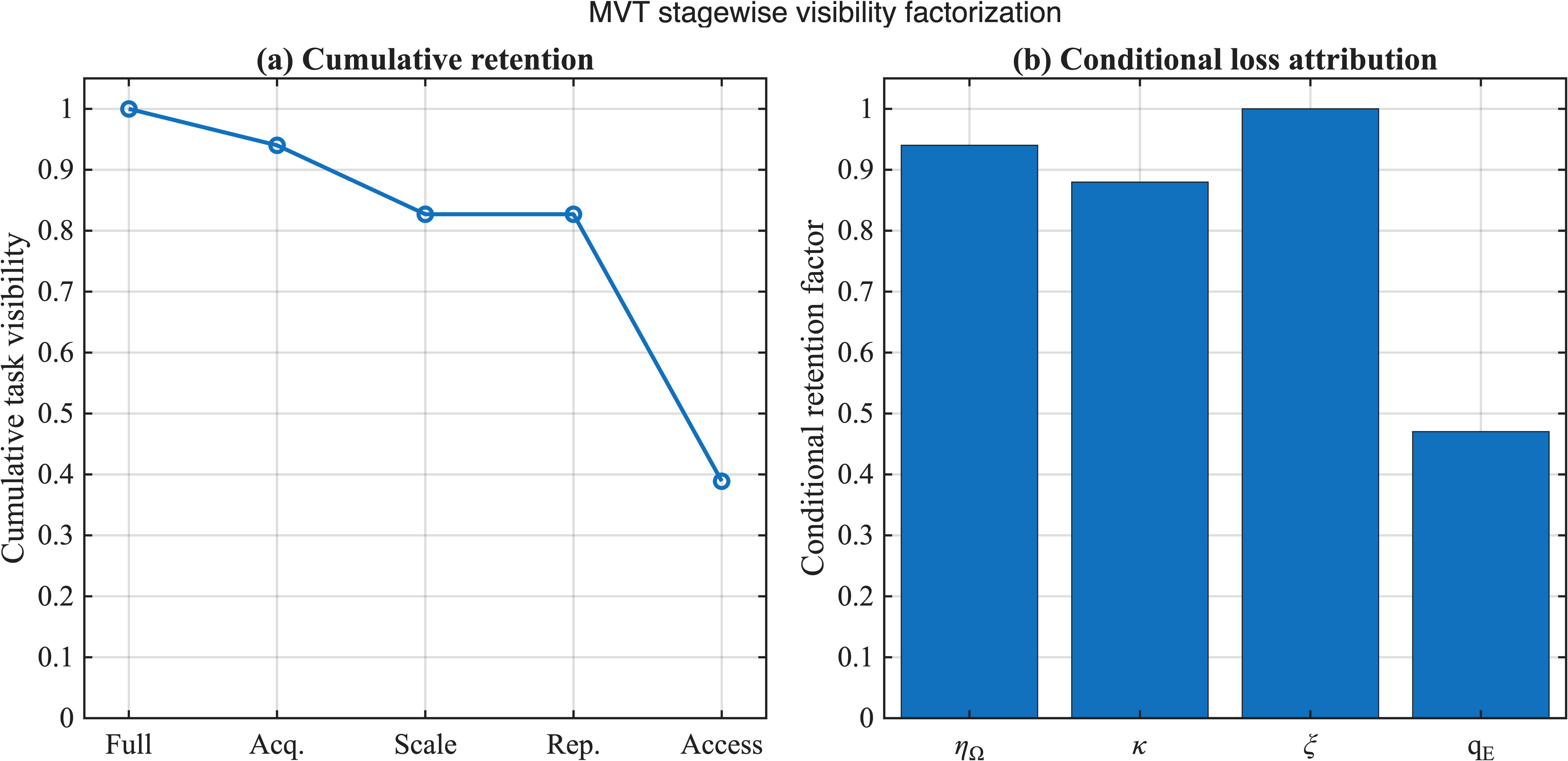}
\caption{
Exact stagewise visibility factorization for the controlled observation chain.
(a) Cumulative task visibility computed from the stage projectors in
\eqref{eq:cumulative_stage_visibility}.
(b) Conditional retention factors defined by \eqref{eq:stage_factors}.
The representation stage introduces no additional geometric loss for the
surviving task component in this example, whereas retained access produces
the largest conditional reduction. The product of the conditional factors
equals the direct final visibility in accordance with
Proposition~\ref{prop:stagewise_factorization}.
}
\label{fig:stagewise}
\end{figure*}

\section{Statistical Visibility and Decision Sufficiency}
\label{sec:statistical_visibility}

Geometric visibility is defined entirely by the prescribed task and the observation-accessible subspace; no probability model is required. Statistical visibility is introduced only after an uncertainty model has been specified. This distinction is fundamental to MVT because geometric accessibility and statistical detectability answer different questions: a task direction may remain accessible through the observation architecture while the retained evidence is statistically too weak to satisfy a prescribed decision requirement.

Projection and subspace geometry have also been connected to decision making in other problem settings. In projection-based fault diagnosis, Ding, Li, and Liu represent system behavior through Hilbert-space subspaces, construct residuals from orthogonal projection relative to nominal system behavior, and use subspace-distance and gap-metric concepts to support fault detection and isolation \cite{DingLiLiu2026ProjectionFaultDiagnosis}. The mathematical use of projection is therefore not itself the novelty claimed by MVT. The evaluated object and the sequence of questions are different. Projection-based fault diagnosis asks whether measured system behavior is sufficiently separated from nominal or fault-associated system subspaces. MVT instead begins with an independently prescribed task direction or task subspace and asks how much of that task remains geometrically accessible through a specified observation architecture. Statistical separation and decision sufficiency are then introduced as distinct, model-dependent layers. The contribution developed in this section is consequently the explicit task-relative connection
between deterministic accessibility, covariance-normalized statistical separation, and operating-point decision sufficiency, without treating these quantities as interchangeable.

Consider the mean-shift observation model
\begin{align}
\mathcal H_0:\quad x &= n, \nonumber\\
\mathcal H_1:\quad x &= \alpha s+n,
\label{eq:mean_shift_model}
\end{align}
where $s\neq0$ is the prescribed task direction, $\alpha$ is its amplitude, and $n$ is a zero-mean random disturbance with covariance $\Sigma_n$. For observation configuration $\mathcal A$,
\begin{equation}
y
=
A_{\mathcal A}x.
\label{eq:observed_model}
\end{equation}
Equation~\eqref{eq:observed_model} places the statistical model after the same deterministic observation architecture used in the geometric formulation. Thus, the task definition and accessible geometry are specified before statistical detectability is introduced.

The observed mean difference and covariance are
\begin{equation}
\Delta\mu_{\mathcal A}
=
\alpha A_{\mathcal A}s,
\qquad
\Sigma_{\mathcal A}
=
A_{\mathcal A}
\Sigma_n
A_{\mathcal A}^{H}.
\label{eq:observed_statistics}
\end{equation}
Equation~\eqref{eq:observed_statistics} shows that both the task displacement and the disturbance are transformed by the same observation operator.

For a common-covariance mean-shift problem, define the covariance-normalized statistical separation as
\begin{equation}
D_{\mathcal A}
=
\Delta\mu_{\mathcal A}^{H}
\Sigma_{\mathcal A}^{\dagger}
\Delta\mu_{\mathcal A},
\label{eq:Dgeneral}
\end{equation}
where $(\cdot)^\dagger$ denotes the Moore--Penrose pseudoinverse. Equation~\eqref{eq:Dgeneral} measures separation on the supported covariance range. It is not introduced as a universal classifier-performance index; its interpretation is conditional on the specified statistical model.

\begin{lemma}[Row-Space Projector Identity]
\label{lem:row_projector}
For any matrix $A\in\mathbb F^{m\times n}$,
\begin{equation}
A^{H}
(AA^{H})^{\dagger}
A
=
P_{\mathcal R(A^{H})}.
\label{eq:row_projector_identity}
\end{equation}
\end{lemma}

\begin{proof}
Let the compact singular-value decomposition of $A$ be
\begin{equation}
A
=
U_r\Sigma_rV_r^{H},
\label{eq:compact_svd}
\end{equation}
where $r=\operatorname{rank}(A)$ and $\Sigma_r$ contains the positive singular values. Then
\begin{align}
A^{H}(AA^{H})^{\dagger}A
&=
V_r\Sigma_rU_r^{H}
\left(
U_r\Sigma_r^2U_r^{H}
\right)^{\dagger}
U_r\Sigma_rV_r^{H}
\nonumber\\
&=
V_rV_r^{H}.
\label{eq:projector_svd_proof}
\end{align}
Since the columns of $V_r$ form an orthonormal basis for $\mathcal R(A^{H})$, $V_rV_r^{H}$ is the orthogonal projector onto that row space, which proves \eqref{eq:row_projector_identity}.
\end{proof}

Lemma~\ref{lem:row_projector} provides the direct mathematical link between covariance-normalized separation and the geometric MVT projector when the pre-observation disturbance is isotropic.

\begin{theorem}[Geometric--Statistical Visibility Relation]
\label{thm:geometric_statistical}
Assume the pre-observation disturbance is isotropic with covariance
\begin{equation}
\Sigma_n
=
\sigma^2 I,
\qquad
\sigma^2>0,
\label{eq:white_covariance}
\end{equation}
and that the same deterministic operator $A_{\mathcal A}$ acts on the task displacement and the disturbance. Then the statistical separation in \eqref{eq:Dgeneral} satisfies
\begin{equation}
D_{\mathcal A}(s)
=
\frac{|\alpha|^2\|s\|^2}{\sigma^2}
q_{\mathcal A}(s).
\label{eq:Dwhite}
\end{equation}
\end{theorem}

\begin{proof}
Under \eqref{eq:white_covariance},
\begin{equation}
\Sigma_{\mathcal A}
=
\sigma^2
A_{\mathcal A}
A_{\mathcal A}^{H},
\label{eq:white_output_covariance}
\end{equation}
and the observed mean difference is
$\Delta\mu_{\mathcal A}=\alpha A_{\mathcal A}s$. Substituting these expressions into \eqref{eq:Dgeneral} gives

\begin{align}
D_{\mathcal A}(s)
&=
\frac{|\alpha|^2}{\sigma^2}
s^{H}
A_{\mathcal A}^{H}
\left(
A_{\mathcal A}A_{\mathcal A}^{H}
\right)^{\dagger}
A_{\mathcal A}s
\nonumber\\
&=
\frac{|\alpha|^2}{\sigma^2}
s^{H}
P_{\mathcal A}
s,
\label{eq:D_projection}
\end{align}
where Lemma~\ref{lem:row_projector} was used in the second line. Since $P_{\mathcal A}$ is an orthogonal projector,
\begin{equation}
s^{H}P_{\mathcal A}s
=
\|P_{\mathcal A}s\|^2
=
\|s\|^2 q_{\mathcal A}(s).
\label{eq:projection_visibility_relation}
\end{equation}
Substitution of \eqref{eq:projection_visibility_relation} into \eqref{eq:D_projection} yields \eqref{eq:Dwhite}.
\end{proof}

Theorem~\ref{thm:geometric_statistical} identifies the precise conditions under which statistical separation factorizes into an intrinsic task-strength term and an observation-induced geometric term. Defining
\begin{equation}
D_{\rm intrinsic}
=
\frac{|\alpha|^2\|s\|^2}{\sigma^2},
\label{eq:D_intrinsic}
\end{equation}
Equation~\eqref{eq:Dwhite} becomes
\begin{equation}
D_{\mathcal A}(s)
=
D_{\rm intrinsic}
q_{\mathcal A}(s).
\label{eq:D_factorized}
\end{equation}
Equation~\eqref{eq:D_factorized} shows that, under the stated white pre-observation noise model, the statistical separation available after observation is the intrinsic separation multiplied by the fraction of task geometry retained by the observation-accessible subspace. The singular-value gains of the deterministic observation operator do not appear explicitly because the same operator acts on both the task displacement and the disturbance covariance. This relation should therefore not be interpreted outside the assumptions stated in Theorem~\ref{thm:geometric_statistical}.

For colored pre-observation covariance $\Sigma_n\succ0$, the corresponding geometry can be expressed in whitened coordinates. Define
\begin{equation}
\widetilde s
=
\Sigma_n^{-1/2}s,
\qquad
B_{\mathcal A}
=
A_{\mathcal A}\Sigma_n^{1/2}.
\label{eq:whitening}
\end{equation}
Equation~\eqref{eq:whitening} expresses the task and observation operator in coordinates in which the pre-observation covariance becomes the identity.

Using
\begin{equation}
A_{\mathcal A}s
=
B_{\mathcal A}\widetilde s
\label{eq:whitened_mean}
\end{equation}
and
\begin{equation}
\Sigma_{\mathcal A}
=
B_{\mathcal A}
B_{\mathcal A}^{H},
\label{eq:whitened_covariance}
\end{equation}
Lemma~\ref{lem:row_projector} gives
\begin{equation}
D_{\mathcal A}(s)
=
|\alpha|^2
\left\|
P_{\mathcal R(B_{\mathcal A}^{H})}
\widetilde s
\right\|^2.
\label{eq:Dcolored}
\end{equation}
Equation~\eqref{eq:Dcolored} shows that colored-noise statistical visibility again admits a geometric interpretation after covariance whitening. However, the relevant geometry is now defined by the whitened task $\widetilde s$ and the whitened observation operator $B_{\mathcal A}$ rather than by $s$ and $A_{\mathcal A}$ in the original Euclidean metric.

A normalized whitened visibility may therefore be written as
\begin{equation}
q_{\mathcal A}^{(\Sigma)}(s)
=
\frac{
\left\|
P_{\mathcal R(B_{\mathcal A}^{H})}
\widetilde s
\right\|^2
}{
\|\widetilde s\|^2
},
\label{eq:q_colored}
\end{equation}
for which
\begin{equation}
D_{\mathcal A}(s)
=
|\alpha|^2
\|\widetilde s\|^2
q_{\mathcal A}^{(\Sigma)}(s).
\label{eq:Dcolored_factorized}
\end{equation}
Equations~\eqref{eq:q_colored} and \eqref{eq:Dcolored_factorized} clarify why geometric visibility in the original Euclidean metric and covariance-weighted statistical visibility need not coincide when the disturbance is anisotropic. The whitening construction preserves the distinction central to MVT: the deterministic accessible geometry is one layer, while the uncertainty-weighted statistical geometry is another.

\subsection{Decision Sufficiency}

Statistical separation becomes decision-relevant only after a required operating point is specified. For the equal-covariance Gaussian mean-shift model with the optimal linear likelihood-ratio statistic, the detection probability can be written as
\begin{equation}
P_D
=
Q\!\left(
Q^{-1}(P_{FA})
-
\sqrt{D}
\right),
\label{eq:PD}
\end{equation}
where $P_{FA}$ is the false-alarm probability and $Q(\cdot)$ is the standard Gaussian upper-tail function. Equation~\eqref{eq:PD} maps covariance-normalized separation to a receiver-operating point under the stated Gaussian assumptions.

For prescribed $P_{FA}$ and target detection probability $P_D^\star$, the required statistical separation is

\begin{equation}
D
\ge
\Gamma_D,
\label{eq:decision_requirement}
\end{equation}
where
\begin{equation}
\Gamma_D
=
\left[
Q^{-1}(P_{FA})
-
Q^{-1}(P_D^\star)
\right]^2.
\label{eq:gamma}
\end{equation}

Equation~\eqref{eq:gamma} defines the statistical separation required by the specified operating point. A task can therefore be geometrically visible, with $q_{\mathcal A}(s)>0$, yet remain decision-invisible if
\begin{equation*}
D_{\mathcal A}(s)<\Gamma_D.
\end{equation*}

Conversely, decision sufficiency is not an intrinsic property of a task direction alone; it depends jointly on task strength, observation geometry, disturbance statistics, and the prescribed operating requirement.

Under the isotropic white-noise assumptions of Theorem~\ref{thm:geometric_statistical}, the decision requirement can be expressed directly in geometric form. Combining \eqref{eq:D_factorized} and \eqref{eq:decision_requirement} gives

\begin{equation}
q_{\mathcal A}(s)
\ge
q_{\rm crit},
\label{eq:q_decision_requirement}
\end{equation}
where
\begin{equation}
q_{\rm crit}
=
\frac{\Gamma_D}
{D_{\rm intrinsic}}.
\label{eq:qcrit_general}
\end{equation}

Equation~\eqref{eq:qcrit_general} converts the required statistical operating point into a minimum geometric visibility requirement for the specified white-input-noise model. This conversion is model dependent and does not turn geometric visibility itself into a statistical probability or decision metric.

The resulting MVT hierarchy is therefore

\begin{equation}
\begin{aligned}
\text{task-accessible geometry}
&\longrightarrow
q_{\mathcal A}(s)
\longrightarrow
D_{\mathcal A}(s)
\\[1mm]
&\longrightarrow
\bigl[D_{\mathcal A}(s)\ge\Gamma_D\bigr].
\end{aligned}
\label{eq:geometry_decision_chain}
\end{equation}

Equation~\eqref{eq:geometry_decision_chain} summarizes the logical progression from deterministic task accessibility to statistical separation and finally to decision-level sufficiency. The stages are connected under explicit assumptions but remain conceptually distinct. Geometric visibility is defined without a probability model; $D_{\mathcal A}$ additionally requires an uncertainty model; and $\Gamma_D$ additionally requires a specified operating criterion. This layered construction is the relevant distinction from projection-based decision frameworks that use projection residuals directly as the geometric basis for system-state discrimination: MVT treats projection as the geometric foundation for a prescribed task and explicitly separates what is accessible, what is statistically distinguishable, and what is sufficient for the required decision.

\section{Intrinsic Multiscale Visibility}
\label{sec:multiscale_visibility}

The multiscale component of MVT is defined through a parameterized family of observation-accessible subspaces rather than through transform magnitude or coefficient energy. Let $\beta\in\mathcal B$ index a family of effective observation operators $A_\beta$. The corresponding accessible-subspace path is
\begin{equation}
\beta
\longmapsto
\mathcal O(\beta)
=
\mathcal R(A_\beta^H),
\label{eq:subspace_path}
\end{equation}
with orthogonal projector
\begin{equation}
P_\beta
=
P_{\mathcal O(\beta)}.
\label{eq:scale_projector}
\end{equation}
Equations~\eqref{eq:subspace_path} and \eqref{eq:scale_projector} make scale dependence an operator-geometric property. The parameter $\beta$ need not represent physical frequency exclusively; it may index any prescribed scale-dependent observation family for which the effective operator is explicitly defined.

For a nonzero prescribed task direction $s$, scale-dependent geometric visibility is
\begin{equation}
q(\beta;s)
=
\frac{\|P_\beta s\|^2}
{\|s\|^2}.
\label{eq:qbeta}
\end{equation}
Equation~\eqref{eq:qbeta} evaluates how strongly the prescribed task aligns with the accessible subspace at scale $\beta$. It is therefore distinct from the magnitude or energy of a transform coefficient at that scale.

For a task subspace $\mathcal S$ with orthonormal basis $U_S$, define
\begin{equation}
q_{\min}(\beta;\mathcal S)
=
\lambda_{\min}
\left(
U_S^H P_\beta U_S
\right).
\label{eq:qmin_beta}
\end{equation}

Equation~\eqref{eq:qmin_beta} gives the least visible unit-norm direction of the prescribed task model at scale $\beta$. Consequently, MVT does not assume that a finer scale, a larger observation rank, or a larger number of measurements necessarily increases task visibility; the determining quantity is alignment with the prescribed task.

\subsection{Projector Distance and Visibility Stability}

To quantify the geometric change between two accessible subspaces $\mathcal O_1$ and $\mathcal O_2$ with orthogonal projectors $P_1$ and $P_2$, define the projector distance

\begin{equation}
d_{\rm proj}
\left(
\mathcal O_1,\mathcal O_2
\right)
=
\|P_1-P_2\|_2 .
\label{eq:dproj}
\end{equation}

Equation~\eqref{eq:dproj} measures the largest directional change induced by replacing one accessible subspace with the other. The definition remains valid when the two subspaces have different dimensions, although in that case the distance can attain its maximum value of one.

Projection-based subspace distances and gap-metric constructions are
established tools in operator theory and robust-system analysis
\cite{Kato1995,GeorgiouSmith1990,Vinnicombe2000}
and are not claimed here as new. For example, gap-metric techniques have been used for fault-detection
performance analysis and fault isolation \cite{LiDing2020GapMetric}, while
Ding, Li, and Liu subsequently formulate a broader projection-based fault
diagnosis paradigm using Hilbert-space system subspaces,
orthogonal-projection residuals, and gap-metric geometry
\cite{DingLiLiu2026ProjectionFaultDiagnosis}. The role of subspace geometry in MVT is different: the subspaces are generated by a parameterized observation family, and their motion is evaluated relative to an independently prescribed task. Accordingly, $d_{\rm proj}$ is used here not as a fault-separation metric, but as a geometric control on how much task-relative visibility can change as the observation scale or configuration varies. The MVT-specific contribution in this section is therefore the coupling of subspace motion to task visibility and the resulting stability and variation bounds along a scale-indexed observation path.

\begin{theorem}[Stability of Single-Task Visibility]
\label{thm:visibility_stability}
Let $\mathcal O_1$ and $\mathcal O_2$ be two accessible subspaces with orthogonal projectors $P_1$ and $P_2$. For every nonzero task direction $s$,
\begin{equation}
\left|
q_{\mathcal O_1}(s)
-
q_{\mathcal O_2}(s)
\right|
\le
d_{\rm proj}
\left(
\mathcal O_1,\mathcal O_2
\right).
\label{eq:stability}
\end{equation}
\end{theorem}

\begin{proof}
From the definition of geometric visibility,
\begin{align}
q_{\mathcal O_1}(s)
-
q_{\mathcal O_2}(s)
&=
\frac{
s^H(P_1-P_2)s
}{
\|s\|^2
}.
\label{eq:stability_difference}
\end{align}
Taking absolute values and applying the Rayleigh-quotient bound gives
\begin{equation}
\left|
s^H(P_1-P_2)s
\right|
\le
\|P_1-P_2\|_2
\|s\|^2 .
\label{eq:rayleigh_bound}
\end{equation}
Substitution of \eqref{eq:rayleigh_bound} into
\eqref{eq:stability_difference} yields \eqref{eq:stability}.
\end{proof}

Theorem~\ref{thm:visibility_stability} provides an intrinsic stability guarantee: task visibility cannot change by more than the corresponding movement of the accessible subspace in projector norm. The result depends only on subspace geometry and does not require a statistical noise model \cite{horn2013matrix,daviskahan1970,stewart1990}.

The same argument extends to the worst-direction visibility of a task subspace.

\begin{proposition}[Stability of Task-Subspace Visibility]
\label{prop:qmin_stability}
For a fixed task subspace $\mathcal S$ with orthonormal basis $U_S$,
\begin{equation}
\left|
q_{\min}(\mathcal O_1;\mathcal S)
-
q_{\min}(\mathcal O_2;\mathcal S)
\right|
\le
\|P_1-P_2\|_2 .
\label{eq:qmin_stability}
\end{equation}
\end{proposition}

\begin{proof}
Define
\begin{equation}
G_j
=
U_S^H P_j U_S,
\qquad
j\in\{1,2\}.
\label{eq:G_two_scales}
\end{equation}
By the Hermitian eigenvalue perturbation bound,
\begin{equation}
\left|
\lambda_{\min}(G_1)
-
\lambda_{\min}(G_2)
\right|
\le
\|G_1-G_2\|_2 .
\label{eq:weyl_qmin}
\end{equation}
Since $U_S$ has orthonormal columns,
\begin{equation}
\|G_1-G_2\|_2
=
\left\|
U_S^H(P_1-P_2)U_S
\right\|_2
\le
\|P_1-P_2\|_2.
\label{eq:compressed_projector_bound}
\end{equation}
Equations~\eqref{eq:weyl_qmin} and
\eqref{eq:compressed_projector_bound} establish
\eqref{eq:qmin_stability}.
\end{proof}

Proposition~\ref{prop:qmin_stability} shows that the least-visible direction of an entire task model is controlled by the same projector geometry that governs single-task visibility.

\subsection{Visibility Variation Along a Scale Path}

For a finite ordered scale grid
\begin{equation}
\beta_0
<
\beta_1
<
\cdots
<
\beta_N,
\label{eq:scale_grid}
\end{equation}
define the discrete projector-path length
\begin{equation}
L_{\rm proj}
=
\sum_{k=1}^{N}
\left\|
P_{\beta_k}
-
P_{\beta_{k-1}}
\right\|_2 .
\label{eq:projector_path_length}
\end{equation}
Equation~\eqref{eq:projector_path_length} measures the accumulated movement of the accessible subspace along the selected scale grid.

The corresponding total variation of single-task visibility is
\begin{equation}
{\rm TV}_q
=
\sum_{k=1}^{N}
\left|
q(\beta_k;s)
-
q(\beta_{k-1};s)
\right|.
\label{eq:q_total_variation}
\end{equation}
Applying Theorem~\ref{thm:visibility_stability} to each adjacent pair yields
\begin{equation}
{\rm TV}_q
\le
L_{\rm proj}.
\label{eq:path_variation_bound}
\end{equation}
Equation~\eqref{eq:path_variation_bound} provides a cumulative stability certificate: the total variation of task visibility across the sampled scale path cannot exceed the accumulated projector motion.

A direct consequence is that operator-norm continuity of the projector path implies continuity of task visibility.

\begin{corollary}[Continuity Along a Multiscale Path]
\label{cor:visibility_continuity}
If
\begin{equation}
\|P_{\beta'}-P_\beta\|_2
\longrightarrow
0
\qquad
\text{as }
\beta'\longrightarrow\beta,
\label{eq:projector_continuity}
\end{equation}
then
\begin{equation}
q(\beta';s)
\longrightarrow
q(\beta;s).
\label{eq:q_continuity}
\end{equation}
\end{corollary}

Corollary~\ref{cor:visibility_continuity} does not require monotonicity of the scale path. A continuous observation family may therefore produce nonmonotonic visibility profiles, multiple local maxima, or several disconnected decision-visible regions.

\subsection{Decision-Visible Scale Regions}

When the statistical model of Section~\ref{sec:statistical_visibility} is also specified, each scale can be associated with a statistical separation $D(\beta;s)$. The decision-visible scale set is
\begin{equation}
\mathcal B_D
=
\left\{
\beta\in\mathcal B:
D(\beta;s)
\ge
\Gamma_D
\right\}.
\label{eq:decision_set}
\end{equation}
Equation~\eqref{eq:decision_set} contains the scales at which the prescribed task satisfies the required statistical operating point.

Under the isotropic white-noise assumptions of Theorem~\ref{thm:geometric_statistical}, the same region can be expressed geometrically as
\begin{equation}
\mathcal B_D
=
\left\{
\beta\in\mathcal B:
q(\beta;s)
\ge
q_{\rm crit}
\right\},
\label{eq:decision_set_q}
\end{equation}
where $q_{\rm crit}$ is defined in \eqref{eq:qcrit_general}. Equation~\eqref{eq:decision_set_q} connects the multiscale geometric profile directly to the prescribed decision requirement under that statistical model.

No monotonicity of $q(\beta;s)$ or $D(\beta;s)$ is assumed. Consequently, $\mathcal B_D$ may consist of one interval, several disconnected intervals, isolated regions on a discrete scale grid, or the empty set. A unique ``critical scale'' is therefore meaningful only when additional monotonicity or structural assumptions are available.

The multiscale component of MVT can thus be summarized by the sequence

\begin{equation}
\begin{aligned}
\beta
&\longmapsto
A_\beta
\longmapsto
\mathcal O(\beta)
\longmapsto
P_\beta
\\[1mm]
&\longmapsto
q(\beta;s)
\longmapsto
D(\beta;s).
\end{aligned}
\label{eq:multiscale_chain}
\end{equation}

where the final statistical mapping is included only after an uncertainty model has been specified. Equation~\eqref{eq:multiscale_chain} clarifies the use of the term \emph{multiscale} in MVT: the theory studies task-relative accessibility along a parameterized path of observation subspaces rather than identifying scale with transform magnitude alone.

\section{Erasure Robustness and Collective Visibility}
\label{sec:robust_collective}

The visibility of a prescribed task may depend not only on the nominal observation architecture but also on which measurements, channels, or representation coefficients remain available. MVT therefore distinguishes two related questions. \emph{Erasure robustness} asks how much task visibility is guaranteed under partial loss, whereas \emph{collective visibility} asks how complementary observation branches combine when they are jointly accessible.

\subsection{Erasure Robustness}

Let $A\in\mathbb F^{m\times n}$ denote an observation operator with row indices
\begin{equation}
[m]=\{1,\ldots,m\}.
\label{eq:row_index_set}
\end{equation}
For an erased row set $F\subseteq[m]$, let
\begin{equation}
E=[m]\setminus F
\label{eq:retained_set}
\end{equation}
denote the retained rows, and let $A_E$ be the corresponding retained observation operator. Equation~\eqref{eq:retained_set} separates the physical loss pattern from the observation geometry that remains available after erasure.

For a prescribed task subspace $\mathcal S$ with orthonormal basis $U_S$, define the exact worst-case visibility under at most $r$ erased rows by
\begin{equation}
G_r(\mathcal S)
=
\min_{\substack{
F\subseteq[m]\\
|F|\le r
}}
\lambda_{\min}
\left(
U_S^H
P_{\mathcal R(A_E^H)}
U_S
\right),
\qquad
E=[m]\setminus F.
\label{eq:Gr}
\end{equation}
Equation~\eqref{eq:Gr} evaluates the least task-model visibility over all admissible erasure patterns. It is therefore a geometric robustness measure rather than a count of retained measurements.

The exact sequence is monotone with respect to the erasure budget:
\begin{equation}
1
\ge
G_0(\mathcal S)
\ge
G_1(\mathcal S)
\ge
\cdots
\ge
G_m(\mathcal S)
\ge
0.
\label{eq:Gr_monotone}
\end{equation}
Equation~\eqref{eq:Gr_monotone} follows because increasing $r$ enlarges the family of admissible erasure patterns over which the minimum in \eqref{eq:Gr} is taken.

To obtain computable lower certificates, write the rows of $A$ as $a_i^H$ and define their task-restricted coordinates by
\begin{equation}
z_i
=
U_S^H a_i,
\qquad
i=1,\ldots,m.
\label{eq:zi}
\end{equation}
Equation~\eqref{eq:zi} retains only the component of each observation row that interacts with the prescribed task subspace.

Define the full task-restricted row energy and the largest removable task-restricted contribution as
\begin{equation}
\begin{aligned}
A_{\mathcal S}
&=
\lambda_{\min}\!\left(
\sum_{i=1}^{m}
z_i z_i^H
\right),
\\[1mm]
\Delta_r(\mathcal S)
&=
\max_{\substack{
F\subseteq[m]\\
|F|\le r
}}
\lambda_{\max}\!\left(
\sum_{i\in F}
z_i z_i^H
\right).
\end{aligned}
\label{eq:restricted_energy}
\end{equation}

The first quantity in \eqref{eq:restricted_energy} measures the weakest direction of the complete task-restricted row system, whereas $\Delta_r(\mathcal S)$ measures the largest spectral contribution that can be removed by at most $r$ erasures.

\begin{theorem}[Task-Restricted Spectral Erasure Certificate]
\label{thm:erasure_certificate}
Let $A\neq 0$. For every erasure budget
$r\in\{0,\ldots,m\}$,
\begin{equation}
G_r(\mathcal S)
\ge
L_r^{\rm spec}(\mathcal S),
\label{eq:erasure_certificate}
\end{equation}

where

\begin{equation}
L_r^{\rm spec}(\mathcal S)
=
\left[
\frac{
A_{\mathcal S}
-
\Delta_r(\mathcal S)
}{
\|A\|_2^2
}
\right]_+,
\label{eq:Lspec}
\end{equation}
and $[x]_+=\max\{x,0\}$.
\end{theorem}

\begin{proof}
Fix an admissible erased set $F$ with $|F|\le r$ and retained set
$E=[m]\setminus F$.

If $A_E=0$, then
$\mathcal R(A_E^H)=\{0\}$ and hence
\[
P_{\mathcal R(A_E^H)}=0.
\]
Therefore, the retained task visibility is zero. Moreover, since
$A_E=0$, all nonzero row contributions of $A$ belong to the erased set
$F$, and thus
\[
\sum_{i\in F} z_i z_i^H
=
\sum_{i=1}^{m} z_i z_i^H.
\]
Consequently,
$\Delta_r(\mathcal S)\ge A_{\mathcal S}$, so that
$L_r^{\rm spec}(\mathcal S)=0$. Hence
\eqref{eq:erasure_certificate} holds trivially in this case.

Assume henceforth that $A_E\neq0$. Then

\begin{equation}
P_{\mathcal R(A_E^H)}
\succeq
\frac{A_E^H A_E}{\|A_E\|_2^2}
\succeq
\frac{A_E^H A_E}{\|A\|_2^2},
\label{eq:projector_energy_bound}
\end{equation}

and restriction to $\mathcal S$ gives

\begin{equation}
U_S^H
P_{\mathcal R(A_E^H)}
U_S
\succeq
\frac{
U_S^H A_E^H A_E U_S
}{
\|A\|_2^2
}.
\label{eq:restricted_projector_bound}
\end{equation}

Using $z_i=U_S^Ha_i$,
\begin{equation}
U_S^H A_E^H A_E U_S
=
\sum_{i\in E}z_i z_i^H
=
\sum_{i=1}^{m}z_i z_i^H
-
\sum_{i\in F}z_i z_i^H.
\label{eq:retained_gram}
\end{equation}

Weyl's eigenvalue inequality therefore yields

\begin{align}
\lambda_{\min}
\left(
U_S^H A_E^H A_E U_S
\right)
&\ge
A_{\mathcal S}
-
\lambda_{\max}
\left(
\sum_{i\in F}z_i z_i^H
\right)
\nonumber\\
&\ge
A_{\mathcal S}
-
\Delta_r(\mathcal S).
\label{eq:weyl_erasure}
\end{align}

Combining \eqref{eq:restricted_projector_bound} and \eqref{eq:weyl_erasure}, and using nonnegativity of the exact visibility, gives

\begin{equation}
\lambda_{\min}
\left(
U_S^H
P_{\mathcal R(A_E^H)}
U_S
\right)
\ge
\left[
\frac{
A_{\mathcal S}
-
\Delta_r(\mathcal S)
}{
\|A\|_2^2
}
\right]_+.
\label{eq:pattern_certificate}
\end{equation}
Since \eqref{eq:pattern_certificate} holds for every admissible $F$, minimization over all $|F|\le r$ proves \eqref{eq:erasure_certificate}.
\end{proof}

Theorem~\ref{thm:erasure_certificate} provides a sufficient lower guarantee for worst-case task visibility without interpreting the number of retained rows as a surrogate for task preservation. The certificate may be conservative because the normalization by $\|A\|_2^2$ and the worst-case removable contribution are both global bounds.

A more explicit certificate can be obtained from task-restricted row strength and cross-row coherence. Define

\begin{equation}
\nu_{\mathcal S}
=
\max_i
\|z_i\|^2,
\qquad
\mu_{\mathcal S}
=
\max_{i\neq j}
|z_i^H z_j|.
\label{eq:coherence_terms}
\end{equation}

Equation~\eqref{eq:coherence_terms} measures, respectively, the largest task-restricted row energy and the largest pairwise task-restricted inner product.

For an erasure set with $|F|\le r$, define
\begin{equation}
\delta_r
=
\begin{cases}
0, & r=0,\\[1mm]
\nu_{\mathcal S}
+
(r-1)\mu_{\mathcal S},
& r\ge1.
\end{cases}
\label{eq:delta_r}
\end{equation}
The corresponding closed-form lower certificate is
\begin{equation}
L_r(\mathcal S)
=
\left[
\frac{
A_{\mathcal S}
-
\delta_r
}{
\|A\|_2^2
}
\right]_+.
\label{eq:Lr}
\end{equation}
Equation~\eqref{eq:Lr} avoids explicit optimization over erased subsets.

\begin{corollary}[Certificate Hierarchy]
\label{cor:certificate_hierarchy}
The exact robustness and the two lower certificates satisfy
\begin{equation}
G_r(\mathcal S)
\ge
L_r^{\rm spec}(\mathcal S)
\ge
L_r(\mathcal S).
\label{eq:robust_hierarchy}
\end{equation}
\end{corollary}

\begin{proof}
The first inequality is Theorem~\ref{thm:erasure_certificate}. For any erased set $F$ with $|F|\le r$, the nonzero eigenvalues of
\begin{equation}
\sum_{i\in F}z_i z_i^H
\label{eq:erased_outer_sum}
\end{equation}
coincide with those of the corresponding Gram matrix. By the row-sum bound,
\begin{equation}
\lambda_{\max}
\left(
\sum_{i\in F}z_i z_i^H
\right)
\le
\nu_{\mathcal S}
+
(|F|-1)\mu_{\mathcal S}
\le
\delta_r.
\label{eq:coherence_bound}
\end{equation}
Hence
\begin{equation}
\Delta_r(\mathcal S)
\le
\delta_r,
\label{eq:Delta_delta}
\end{equation}
and substitution into \eqref{eq:Lspec} gives
$L_r^{\rm spec}(\mathcal S)\ge L_r(\mathcal S)$.
\end{proof}

Corollary~\ref{cor:certificate_hierarchy} is one-sided. In particular, $L_r(\mathcal S)=0$ or $L_r^{\rm spec}(\mathcal S)=0$ does not imply that the exact robustness $G_r(\mathcal S)$ is zero; it only means that the corresponding lower certificate is inconclusive.

\subsection{Monotonicity Under Retained Information}

The task visibility associated with retained rows is monotone under set inclusion.

\begin{proposition}[Retained-Set Monotonicity]
\label{prop:retained_monotonicity}
If $K_1\subseteq K_2$ are two retained row sets, then
\begin{equation}
\mathcal R(A_{K_1}^H)
\subseteq
\mathcal R(A_{K_2}^H),
\label{eq:rowspace_inclusion}
\end{equation}
and therefore
\begin{equation}
P_{\mathcal R(A_{K_1}^H)}
\preceq
P_{\mathcal R(A_{K_2}^H)}.
\label{eq:proj_monotonicity1}
\end{equation}
Consequently,
\begin{equation}
q_{\min}(A_{K_1};\mathcal S)
\le
q_{\min}(A_{K_2};\mathcal S).
\label{eq:proj_monotonicity}
\end{equation}
\end{proposition}

\begin{proof}
Adding retained rows can only enlarge, or leave unchanged, the span of the rows of the observation operator, which proves \eqref{eq:rowspace_inclusion}. Orthogonal projectors onto nested subspaces are ordered in the Loewner sense, giving \eqref{eq:proj_monotonicity1}. Compressing this inequality to the task subspace with $U_S$ and taking the minimum eigenvalue yields \eqref{eq:proj_monotonicity}.
\end{proof}

Proposition~\ref{prop:retained_monotonicity} supports exact retained-set search and branch-and-bound reasoning because any subset of an already insufficient retained set cannot have larger task visibility. The proposition does not, however, imply universally sub-combinatorial worst-case complexity; exact evaluation of \eqref{eq:Gr} can still require combinatorial subset search.

\subsection{Controlled Erasure Example}

For the controlled comparison, the prescribed task model is
two-dimensional with
\begin{equation}
\mathcal S=\mathbb R^2,
\qquad
U_S=I_2.
\label{eq:erasure_task_model}
\end{equation}

The minimal two-row quadrature observer is
\begin{equation}
A^{(2)}
=
\begin{bmatrix}
1 & 0\\
0 & 1
\end{bmatrix},
\label{eq:observer_two_row}
\end{equation}
whereas the four-row phase-redundant observer is
\begin{equation}
A^{(4)}
=
\begin{bmatrix}
1 & 0\\
1/\sqrt{2} & 1/\sqrt{2}\\
0 & 1\\
-1/\sqrt{2} & 1/\sqrt{2}
\end{bmatrix}.
\label{eq:observer_four_row}
\end{equation}

The rows of $A^{(4)}$ correspond to phase directions
$0^\circ$, $45^\circ$, $90^\circ$, and $135^\circ$.
Every pair of retained rows of $A^{(4)}$ is linearly independent and
therefore spans the complete two-dimensional task model. In contrast,
removing either row of $A^{(2)}$ leaves only a one-dimensional accessible
subspace.

Exact enumeration of the admissible loss patterns gives

\begin{equation}
G_r^{(2)}
=
(1,0,0),
\qquad
G_r^{(4)}
=
(1,1,1,0,0).
\label{eq:erasureseq}
\end{equation}

Equation~\eqref{eq:erasureseq} shows that both observers preserve the
complete task model before erasure, but their behavior under partial loss
is different. The minimal observer loses complete task-model visibility
after one erased row, whereas the phase-redundant observer retains
$G_r(\mathcal S)=1$ under every loss pattern containing at most two
erased rows.

Figure~\ref{fig:erasure} visualizes the exact robustness sequences in \eqref{eq:erasureseq}. The example demonstrates that redundancy can preserve alternative access paths to a prescribed task model without creating information that was absent upstream.

\begin{figure}[t]
\centering
\includegraphics[width=\columnwidth]
{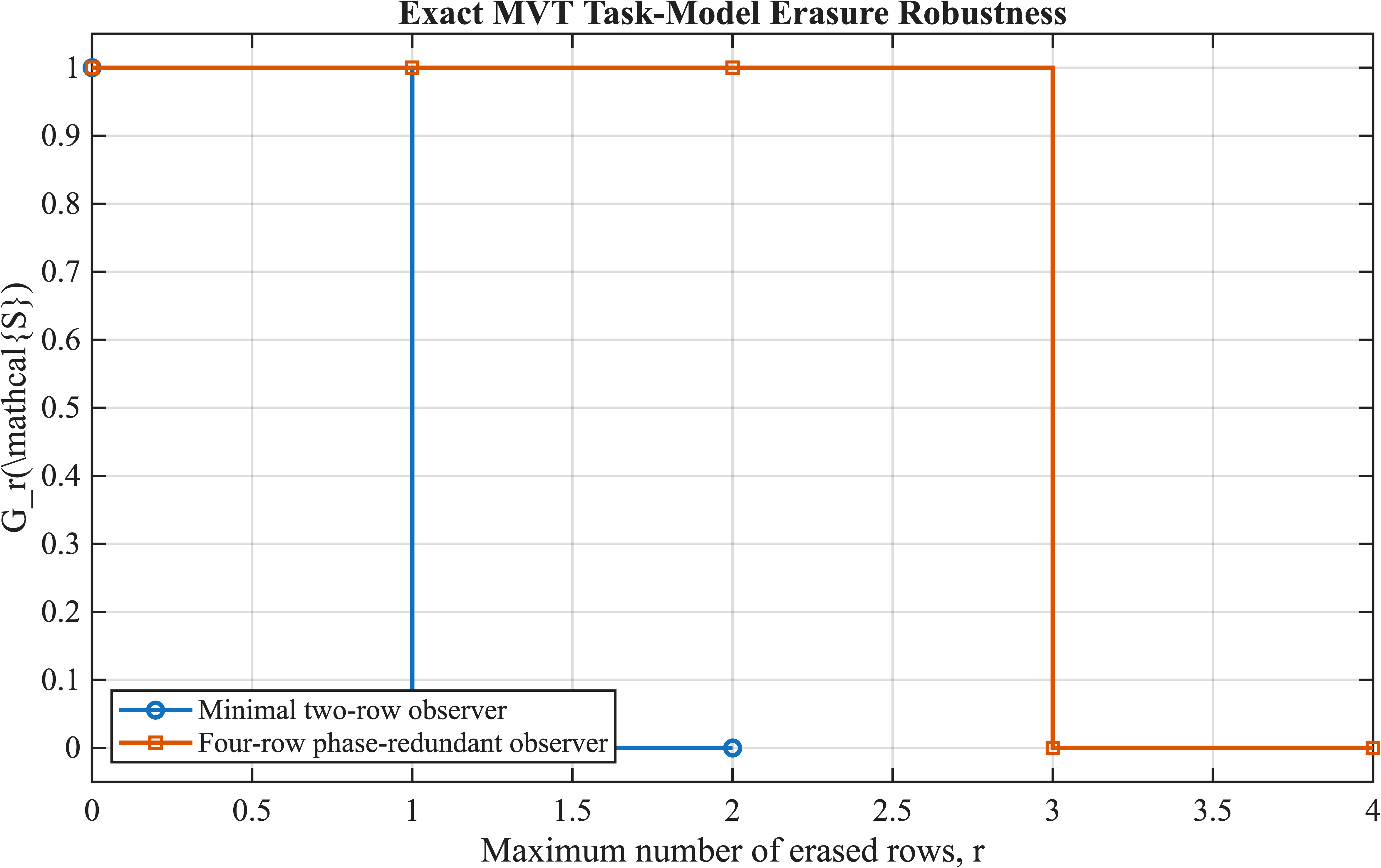}
\caption{
Exact task-model erasure robustness for a minimal two-row observer and a
four-row phase-redundant observer.
The sequences reproduce \eqref{eq:erasureseq}.
The minimal observer loses full task-model visibility after one row erasure,
whereas the redundant observer maintains
$G_r(\mathcal S)=1$ through any two row erasures.
}
\label{fig:erasure}
\end{figure}

The connection with erasure-robust frame design and compressed sensing is structural but not identical. Frame and compressed-sensing theory typically impose ambient-space or reconstruction-oriented conditions, whereas \eqref{eq:Gr} evaluates preservation only on the prescribed task subspace \cite{holmes2004optimal,fickus2012numerically,fickus2014group,candes2005decoding,foucart2013mathematical}.

\subsection{Collective Visibility}

Erasure robustness concerns loss from an existing observation system. Collective visibility addresses the complementary problem: how much task information becomes accessible when several observation branches are jointly available.

Let $\mathcal O_k$ denote the accessible subspace associated with branch $k$, $k=1,\ldots,K$. For a retained branch set $J\subseteq\{1,\ldots,K\}$, define the joint accessible subspace

\begin{equation}
\mathcal O_J
=
\operatorname{span}
\left\{
\mathcal O_k:
k\in J
\right\}.
\label{eq:joint_space}
\end{equation}

Equation~\eqref{eq:joint_space} represents genuine joint access to the selected observation branches.

For a nonzero prescribed task direction $s$, joint geometric visibility is
\begin{equation}
q(J;s)
=
\frac{
\|P_{\mathcal O_J}s\|^2
}{
\|s\|^2
}.
\label{eq:qjoint}
\end{equation}
Equation~\eqref{eq:qjoint} differs from an average or sum of branchwise visibility scores because it is computed after forming the joint accessible subspace.

\begin{proposition}[Monotonicity of Collective Visibility]
\label{prop:collective_monotonicity}
If $J_1\subseteq J_2$, then
\begin{equation}
q(J_1;s)
\le
q(J_2;s).
\label{eq:collective_monotonicity}
\end{equation}
In particular, for an added branch $k\notin J$,
\begin{equation}
\delta q_{k|J}(s)
:=
q(J\cup\{k\};s)
-
q(J;s)
\ge
0.
\label{eq:innovation_nonnegative}
\end{equation}
\end{proposition}

\begin{proof}
From $J_1\subseteq J_2$,
\begin{equation}
\mathcal O_{J_1}
\subseteq
\mathcal O_{J_2}.
\label{eq:collective_nested}
\end{equation}
Hence
\begin{equation}
P_{\mathcal O_{J_1}}
\preceq
P_{\mathcal O_{J_2}}.
\label{eq:collective_projector_order}
\end{equation}
Taking the quadratic form with $s/\|s\|$ gives \eqref{eq:collective_monotonicity}, and the special case $J_2=J\cup\{k\}$ gives \eqref{eq:innovation_nonnegative}.
\end{proof}

The conditional innovation can equivalently be written as
\begin{equation}
\delta q_{k|J}(s)
=
\frac{
s^H
\left(
P_{\mathcal O_{J\cup\{k\}}}
-
P_{\mathcal O_J}
\right)
s
}{
\|s\|^2
}.
\label{eq:innovation}
\end{equation}
Equation~\eqref{eq:innovation} is zero when the added branch provides no new task-relevant accessible component relative to the current joint space. A duplicate branch therefore has zero innovation, but zero innovation can also occur when the added branch enlarges the ambient accessible space only in directions orthogonal to the prescribed task.

For a fixed retained branch count $b$, define the worst, mean, and best collective visibilities by
\begin{align}
V_b^-
&=
\min_{|J|=b}
q(J;s),
\nonumber\\[1mm]
\overline V_b
&=
\binom{K}{b}^{-1}
\sum_{|J|=b}
q(J;s),
\nonumber\\[1mm]
V_b^+
&=
\max_{|J|=b}
q(J;s).
\label{eq:collective_envelopes}
\end{align}
Equation~\eqref{eq:collective_envelopes} separates retained-count effects from retained-pattern effects. Two subsets with the same cardinality can therefore have different task visibility if their accessible subspaces have different alignment with $s$.

The pattern-sensitivity width is
\begin{equation}
\Psi_b
=
V_b^+
-
V_b^-.
\label{eq:Psi}
\end{equation}
Equation~\eqref{eq:Psi} quantifies how strongly collective visibility depends on which $b$ branches are retained. A value $\Psi_b=0$ means that every retained set of cardinality $b$ produces the same task visibility.

Under a specified statistical model, collective geometry can also produce collective decision sufficiency. A branch set $J$ is said to exhibit purely collective decision visibility when
\begin{equation}
D_k
<
\Gamma_D
\qquad
\forall k\in J,
\label{eq:individual_subthreshold}
\end{equation}
while
\begin{equation}
D_J
\ge
\Gamma_D.
\label{eq:collective_threshold}
\end{equation}
Equations~\eqref{eq:individual_subthreshold} and \eqref{eq:collective_threshold} distinguish local statistical insufficiency from joint sufficiency. The condition does not imply that the individual branches contain no task information; rather, each branch is individually insufficient for the prescribed operating point while their genuinely joint accessible geometry is sufficient.

The collective component of MVT can therefore be summarized by
\begin{equation}
\{\mathcal O_k\}_{k=1}^{K}
\longrightarrow
\mathcal O_J
\longrightarrow
q(J;s)
\longrightarrow
\delta q_{k|J}(s)
\longrightarrow
D_J,
\label{eq:collective_chain}
\end{equation}
where the final statistical quantity is introduced only after an uncertainty model has been specified. Equation~\eqref{eq:collective_chain} separates geometric complementarity from statistical and decision-level interpretation.

\section{Computational Verification Protocol}
\label{sec:computational_verification}

The numerical study is designed to verify the finite-precision implementation of the analytical relations established in the preceding sections. Numerical agreement is not treated as a substitute for proof. The evidence is therefore organized in four levels: analytical results, numerical consistency checks, controlled signal experiments with known task structure, and subsequent validation on measured engineering data.

\begin{figure}[t]
\centering

\begin{tikzpicture}[
    >=Latex,
    node distance=3.0mm,
    every node/.style={font=\scriptsize},
    box/.style={
        draw,
        rounded corners,
        align=center,
        text width=4.35cm,
        minimum height=0.90cm,
        line width=0.8pt
    },
    arr/.style={
        -{Latex[length=1.8mm]},
        thick
    }
]

\node[box] (theory)
{
\textbf{Analytical theory}\\
Definitions, propositions, theorems,
bounds, and proofs
};

\node[box,below=of theory] (numerical)
{
\textbf{Finite-precision verification}\\
Independent numerical evaluation of
the established analytical relations
};

\node[box,below=of numerical] (signal)
{
\textbf{Controlled signal study}\\
Evaluation under prescribed task,
signal, noise, and observation conditions
};

\node[box,below=of signal] (measured)
{
\textbf{Measured-data validation}\\
Application-specific assessment on
real engineering observations
};

\draw[arr] (theory) -- (numerical);
\draw[arr] (numerical) -- (signal);
\draw[arr] (signal) -- (measured);

\end{tikzpicture}

\caption{
Evidence hierarchy used for MVT evaluation.
Analytical results establish the mathematical properties of the framework;
finite-precision verification tests their numerical implementation;
controlled signal studies evaluate behavior under known task and observation
conditions; and measured-data validation is reserved for subsequent
application-specific studies.
}
\label{fig:validation_hierarchy}

\end{figure}
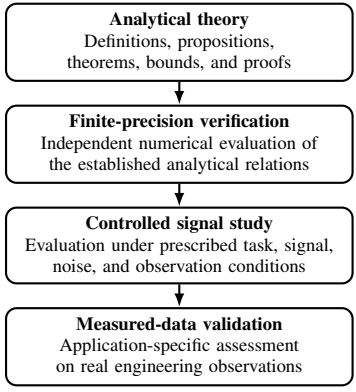

As summarized in Fig.~\ref{fig:validation_hierarchy}, the present paper separates mathematical validity, numerical implementation, and empirical validation. The first two levels test whether the proposed quantities and derived relations are internally consistent, whereas the controlled signal study examines their interpretation under a prescribed benchmark. The final level is not claimed in the present study.

\subsection{Numerical Realization}

All paper-level computations were performed in MATLAB R2025b Update~1. The nominal audit tolerance was
\begin{equation}
\varepsilon_{\rm audit}
=
10^{-10}.
\label{eq:audit_tolerance}
\end{equation}
Equation~\eqref{eq:audit_tolerance} is used only as a numerical acceptance threshold for the registered consistency checks; it is not a statistical significance level.

For each effective observation operator $A$, the observation-accessible
subspace was reconstructed from a compact singular-value decomposition,
\begin{equation}
A
=
U_r\Sigma_rV_r^H,
\label{eq:numerical_svd}
\end{equation}
where the retained numerical rank $r$ was determined from the singular
values $\{\sigma_i\}$ using
\begin{equation}
\tau_{\rm rank}
=
\max\!\left\{
\max(m,n)\,\operatorname{eps}(\sigma_{\max}),
\;
10^{-14}\sigma_{\max}
\right\},
\label{eq:rank_tolerance}
\end{equation}
with
\begin{equation}
r
=
\#\left\{
i:\sigma_i>\tau_{\rm rank}
\right\}.
\label{eq:numerical_rank}
\end{equation}
Here, $\operatorname{eps}(\sigma_{\max})$ denotes the floating-point
spacing at the largest singular value. This adaptive threshold is used
only to determine the numerical row-space dimension and is distinct
from the paper-level audit tolerance
$\varepsilon_{\rm audit}=10^{-10}$.

The corresponding orthogonal projector was evaluated as

\begin{equation}
P_{\mathcal R(A^H)}
=
V_rV_r^H.
\label{eq:numerical_projector}
\end{equation}

Equation~\eqref{eq:numerical_projector} is the computational form used for the geometric visibility calculations and is consistent with Lemma~\ref{lem:row_projector}.

The numerical procedure applied to each prescribed observation configuration can be summarized as

\begin{equation}
\begin{aligned}
A
&\longrightarrow
P_{\mathcal R(A^H)}
\longrightarrow
q
\\
&\longrightarrow
D
\longrightarrow
\text{robustness / decision quantities}.
\end{aligned}
\label{eq:computational_chain}
\end{equation}

where the statistical quantities are evaluated only when the corresponding uncertainty model is specified. Equation~\eqref{eq:computational_chain} is a computational realization of the MVT framework rather than a separate optimization algorithm.

For a single prescribed task direction, visibility was evaluated from
\begin{equation}
q_A(s)
=
\frac{
s^H
P_{\mathcal R(A^H)}
s
}{
\|s\|^2
},
\label{eq:numerical_q}
\end{equation}
while task-subspace visibility was evaluated through the minimum eigenvalue of
\begin{equation}
U_S^H
P_{\mathcal R(A^H)}
U_S.
\label{eq:numerical_qmin}
\end{equation}
Equations~\eqref{eq:numerical_q} and \eqref{eq:numerical_qmin} reproduce the definitions in Section~\ref{sec:geometric_visibility} through numerically constructed projectors.

The same implementation was used to evaluate stagewise factorization, scale-dependent visibility, projector-path stability, exact erasure robustness, and collective visibility. Exact erasure quantities were obtained by explicit enumeration of the admissible retained-row patterns in the controlled low-dimensional example; therefore, no claim of sub-combinatorial exact robustness computation is made.

\subsection{Registered Numerical Audits}

The verification protocol uses independently computable quantities whenever possible. Examples include comparison of direct and factorized stagewise visibility, numerical and analytical Gaussian visibility profiles, projector-stability inequalities, exact erasure sequences, and collective-visibility reference values.

A registered audit is considered successful when its residual or inequality violation satisfies
\begin{equation}
\varepsilon_j
\le
\varepsilon_{\rm audit}.
\label{eq:audit_rule}
\end{equation}
Equation~\eqref{eq:audit_rule} provides a common finite-precision acceptance rule for the numerical checks reported in Table~\ref{tab:audit}.

\begin{table}[t]
\caption{Paper-Level Numerical Audit}
\label{tab:audit}
\centering
\scriptsize
\begin{tabular}{p{3.55cm}cc}
\toprule
Check & Residual / violation & Status\\
\midrule
Stagewise direct vs.\ product
& $0$
& Pass\\

Gaussian frequency reference
& $9.992\times10^{-15}$
& Pass\\

Task maximum at 260 Hz
& $0$
& Pass\\

Decision interval lower bound
& $0$
& Pass\\

Decision interval upper bound
& $0$
& Pass\\

Projector-stability inequality
& $0$
& Pass\\

Minimal-observer erasure sequence
& $0$
& Pass\\

Redundant-observer erasure sequence
& $0$
& Pass\\

$q_{258}$ reference
& $2.944\times10^{-13}$
& Pass\\

$q_{262}$ reference
& $3.064\times10^{-13}$
& Pass\\

$q_{12}$ reference
& $1.591\times10^{-13}$
& Pass\\

Joint decision threshold
& $0$
& Pass\\
\midrule
Registered checks passed
& \multicolumn{2}{c}{12/12}\\
\bottomrule
\end{tabular}
\end{table}

As reported in Table~\ref{tab:audit}, all 12 registered checks satisfy the tolerance in \eqref{eq:audit_rule}. The largest nonzero residual among the listed reference comparisons is $3.064\times10^{-13}$, which remains below the prescribed audit tolerance.

The near-machine-precision agreement between the numerical projector calculation and the analytical Gaussian frequency reference is
\begin{equation}
\max_f
\left|
q_{\rm num}(f;s)
-
q_{\rm ref}(f;s)
\right|
=
9.992\times10^{-15}.
\label{eq:gaussian_audit_result}
\end{equation}
Equation~\eqref{eq:gaussian_audit_result} verifies the numerical implementation of that controlled analytical reference; it does not constitute independent empirical evidence for arbitrary signals.

The stagewise audit independently evaluates the final visibility through the complete projector and through the conditional product in Proposition~\ref{prop:stagewise_factorization}. The resulting residual is
\begin{equation}
\left|
q_{\rm direct}
-
q_{\rm product}
\right|
=
0
\label{eq:factorization_audit_result}
\end{equation}
at the reported numerical precision. Equation~\eqref{eq:factorization_audit_result} verifies the implemented telescoping factorization for the controlled case.

Likewise, the numerical multiscale sweep satisfies
\begin{equation}
\max_k
\left[
\left|
q(\beta_k;s)
-
q(\beta_{k-1};s)
\right|
-
\left\|
P_{\beta_k}
-
P_{\beta_{k-1}}
\right\|_2
\right]_+
=
0,
\label{eq:stability_audit_result}
\end{equation}
which is consistent with Theorem~\ref{thm:visibility_stability}. The test verifies the implementation over the sampled scale path and does not replace the analytical proof.

For the controlled robustness example, exact enumeration reproduces
\begin{equation}
G_r^{(2)}
=
(1,0,0),
\qquad
G_r^{(4)}
=
(1,1,1,0,0),
\label{eq:audit_erasure_sequences}
\end{equation}
with zero reference residual. Equation~\eqref{eq:audit_erasure_sequences} confirms the numerical retained-set calculations used in Section~\ref{sec:robust_collective}.

The collective-visibility calculations similarly reproduce
\begin{equation}
q_{258}
=
0.320786839766,
\qquad
q_{262}
=
0.320786839766,
\label{eq:audit_collective_single}
\end{equation}
and
\begin{equation}
q_{12}
=
0.581712973374,
\label{eq:audit_collective_joint}
\end{equation}
within the residuals listed in Table~\ref{tab:audit}. These values are subsequently used in the controlled signal analysis rather than treated as independent validation data.

The purpose of Table~\ref{tab:audit} is therefore limited but important: it establishes reproducibility and finite-precision consistency of the implemented MVT calculations for the registered controlled cases. It does not demonstrate universal empirical validity, superiority over established signal representations, or performance on measured engineering data. Those questions require application-specific validation under independently specified task, observation, uncertainty, and decision conditions.

\section{Controlled Synthetic-Signal Study}
\label{sec:controlled_study}

The controlled signal study examines whether the geometric, statistical, robust, and collective quantities introduced in the preceding sections produce consistent and interpretable results when the prescribed task is known exactly. The objective is not to establish empirical superiority over conventional signal representations. Instead, the benchmark is constructed to separate \emph{signal prominence} from \emph{task-relative accessibility} under reproducible conditions.

\subsection{Benchmark Design and Comparison Principle}

The synthetic benchmark is sampled at
$f_s=1000$~Hz over a duration of $10$~s, giving $N=10\,000$ samples.
It contains a persistent 40-Hz sinusoid, a persistent 90-Hz harmonic,
a finite-duration 120--220-Hz chirp, a localized prescribed task,
a stronger time-overlapping nuisance burst, a 180-Hz damped ring-down,
a short broadband impulse, and additive pre-observation white Gaussian
noise with standard deviation $\sigma=0.4$. The two localized components
are centered at $t_0=5.20$~s and use the Gaussian temporal width
$\sigma_t=0.12$~s. The frequency-indexed MVT observation family is evaluated
over 100--450~Hz with a 0.5-Hz spacing. The prescribed task is a comparatively weak component at

\begin{equation}
f_{\rm task}=260~\mathrm{Hz},
\label{eq:task_frequency}
\end{equation}

whereas a stronger localized nuisance is placed at

\begin{equation}
f_{\rm nuisance}=380~\mathrm{Hz}.
\label{eq:nuisance_frequency}
\end{equation}

The realized nuisance-to-task peak-amplitude ratio is $8.463783103$. White Gaussian noise is introduced before observation, so the nominal statistical experiment is consistent with the pre-observation white-noise setting of Theorem~\ref{thm:geometric_statistical}.

The benchmark is deliberately constructed so that the physically dominant component need not coincide with the prescribed task. This distinction permits a controlled test of whether an observation can be dominated by signal energy while retaining a different geometric relationship with the task direction.

FFT, STFT, CWT, HHT, and MVT do not report the same mathematical quantity. Consequently, their characteristic frequencies in Table~\ref{tab:methods} are not treated as competing estimates of a common unknown parameter. FFT reports global Fourier prominence, STFT and CWT identify strong localized representation responses, the implemented HHT analysis provides an adaptive local response, whereas MVT evaluates accessibility relative to the prescribed task. The comparison is therefore interpretive rather than an accuracy ranking.

\begin{table}[t]
\caption{Characteristic Outputs for the Controlled Signal}
\label{tab:methods}
\centering
\scriptsize
\begin{tabular}{p{1.0cm}p{3.15cm}p{1.35cm}}
\toprule
Method & Reported quantity & Frequency\\
\midrule
FFT
& Global maximum of complete-record Fourier magnitude
& 40.000 Hz\\

STFT
& Maximum local STFT response near task time
& 379.883 Hz\\

CWT
& Maximum local CWT response near task time
& 374.577 Hz\\

HHT
& Maximum implemented local amplitude-weighted response
& 335.000 Hz\\

MVT
& Maximum prescribed-task geometric visibility
& 260.000 Hz\\

MVT
& Maximum complete-signal energy within the same observation family
& 380.000 Hz\\
\bottomrule
\end{tabular}
\end{table}

As summarized in Table~\ref{tab:methods}, the FFT maximum is associated with a persistent global component, while the strongest local STFT and CWT responses occur near the stronger localized nuisance. The implemented HHT response is not supplied with the prescribed task direction and is therefore not interpreted as a task-frequency estimator. These results are consistent with the different quantities computed by the respective methods and should not be interpreted as evidence that one representation is universally more accurate than another.

\subsection{Task-Relative Frequency-Indexed Observation Family}

The MVT experiment uses a Gaussian-localized sine/cosine observation bank indexed by frequency $f$ over the interval 100--450~Hz with 0.5-Hz spacing. The observation family is centered at the prescribed local event time and uses temporal width
\begin{equation}
\sigma_t=0.12~\mathrm{s}.
\label{eq:sigma_t}
\end{equation}
For each frequency, let $\mathcal O(f)$ denote the row space generated by the corresponding localized sine/cosine pair.

For the normalized task direction $\|s\|=1$, the experiment evaluates two different quantities from the \emph{same} observation family:
\begin{equation}
\begin{aligned}
q(f;s)
&=
\|P_{\mathcal O(f)}s\|^2,
\\[1mm]
E_x(f)
&=
\|P_{\mathcal O(f)}x\|^2.
\end{aligned}
\label{eq:q_energy}
\end{equation}
Equation~\eqref{eq:q_energy} is central to the controlled experiment. The first quantity measures task-relative geometric visibility, whereas the second measures the projected energy of the complete observed signal. Their difference therefore cannot be attributed to different transforms, frequency grids, or cross-method normalization.

The task-relative visibility reaches
\begin{equation}
q_{\max}
=
1,
\qquad
f_q^\star
=
260~\mathrm{Hz}.
\label{eq:qmax_result}
\end{equation}
Equation~\eqref{eq:qmax_result} indicates exact alignment between the prescribed task and the accessible subspace at the task frequency.

By contrast, the projected complete-signal energy within the same MVT observation family reaches its maximum at
\begin{equation}
f_E^\star
=
380~\mathrm{Hz}.
\label{eq:Emax_result}
\end{equation}
The maxima in \eqref{eq:qmax_result} and \eqref{eq:Emax_result} therefore satisfy
\begin{equation}
\operatorname*{arg\,max}_{f}
E_x(f)
\neq
\operatorname*{arg\,max}_{f}
q(f;s).
\label{eq:keyresult}
\end{equation}
Equation~\eqref{eq:keyresult} is the principal controlled counterexample of this experiment: dominance in observation energy does not imply maximal accessibility of the prescribed task, even when both quantities are computed from the same frequency-indexed observation family.

\subsection{Analytical Visibility Reference}

For the matched Gaussian-localized task and observation family used in this experiment, the analytical task-visibility profile is
\begin{equation}
q_{\rm ref}(f;s)
=
\exp\!\left[
-2\pi^2\sigma_t^2
(f-f_{\rm task})^2
\right].
\label{eq:gaussian_visibility_reference}
\end{equation}
Equation~\eqref{eq:gaussian_visibility_reference} provides an analytical reference for the numerical projector calculation. Its maximum occurs at $f=f_{\rm task}$, where $q_{\rm ref}=1$.

The numerical and analytical profiles agree to
\begin{equation}
\max_f
\left|
q_{\rm num}(f;s)
-
q_{\rm ref}(f;s)
\right|
=
9.992\times10^{-15},
\label{eq:visibility_reference_residual}
\end{equation}
which is consistent with the computational audit reported in Table~\ref{tab:audit}. Equation~\eqref{eq:visibility_reference_residual} verifies the implementation of this controlled analytical case; it is not an empirical performance measure.

\subsection{Detectability and Decision-Visible Frequency Region}

At the nominal noise level, the intrinsic statistical separation is

\begin{equation}
D_{\rm intrinsic}=18.
\label{eq:D_intrinsic_experiment}
\end{equation}
Because $q_{\max}=1$, Theorem~\ref{thm:geometric_statistical} gives
\begin{equation}
D_{\max}
=
18.
\label{eq:Dmax}
\end{equation}

For the prescribed operating point, we set
\begin{equation}
P_{FA}=0.05,
\qquad
P_D^\star=0.90.
\label{eq:operating_point_experiment}
\end{equation}
Using \eqref{eq:gamma}, the corresponding decision threshold is
\begin{equation}
\Gamma_D
=
8.563847350668.
\label{eq:gamma_experiment}
\end{equation}
Combining \eqref{eq:D_intrinsic_experiment} and
\eqref{eq:gamma_experiment} with \eqref{eq:qcrit_general} gives
\begin{equation}
q_{\rm crit}
=
\frac{\Gamma_D}
{D_{\rm intrinsic}}
=
0.475769297259.
\label{eq:qcrit}
\end{equation}

Equation~\eqref{eq:qcrit} converts the statistical operating requirement
into the minimum task-relative geometric visibility required by this
specific white-input-noise model.

On the 0.5-Hz numerical observation grid, the decision-visible frequency
region is

\begin{equation}
258.5~\mathrm{Hz}
\le
f
\le
261.5~\mathrm{Hz}.
\label{eq:grid_interval}
\end{equation}
The continuous decision boundary follows from
\begin{equation}
q_{\rm ref}(f;s)
=
q_{\rm crit}.
\label{eq:continuous_boundary_condition}
\end{equation}

Substituting \eqref{eq:gaussian_visibility_reference} into
\eqref{eq:continuous_boundary_condition} yields

\begin{equation}
|f-f_{\rm task}|
\le
\sqrt{
\frac{
-\ln q_{\rm crit}
}{
2\pi^2\sigma_t^2
}
}.
\label{eq:continuous_visible_condition}
\end{equation}

For $\sigma_t=0.12$~s and the value of $q_{\rm crit}$ in
\eqref{eq:qcrit}, the resulting continuous interval is

\begin{equation}
258.383423~\mathrm{Hz}
\le
f
\le
261.616577~\mathrm{Hz}.
\label{eq:continuous_interval}
\end{equation}

The numerical interval in \eqref{eq:grid_interval} is therefore consistent with the continuous analytical reference in \eqref{eq:continuous_interval} to the expected resolution of the 0.5-Hz frequency grid.

Figure~\ref{fig:spectra} shows the geometric and statistical views of the same observation family. Figure~\ref{fig:spectra}(a) compares the numerical projector calculation with the analytical profile in \eqref{eq:gaussian_visibility_reference}, while Fig.~\ref{fig:spectra}(b) shows the corresponding detectability spectrum and the decision threshold $\Gamma_D$.

\begin{figure*}[t]
\centering
\subfloat[Task-relative visibility spectrum.]{%
\includegraphics[width=0.49\textwidth]
{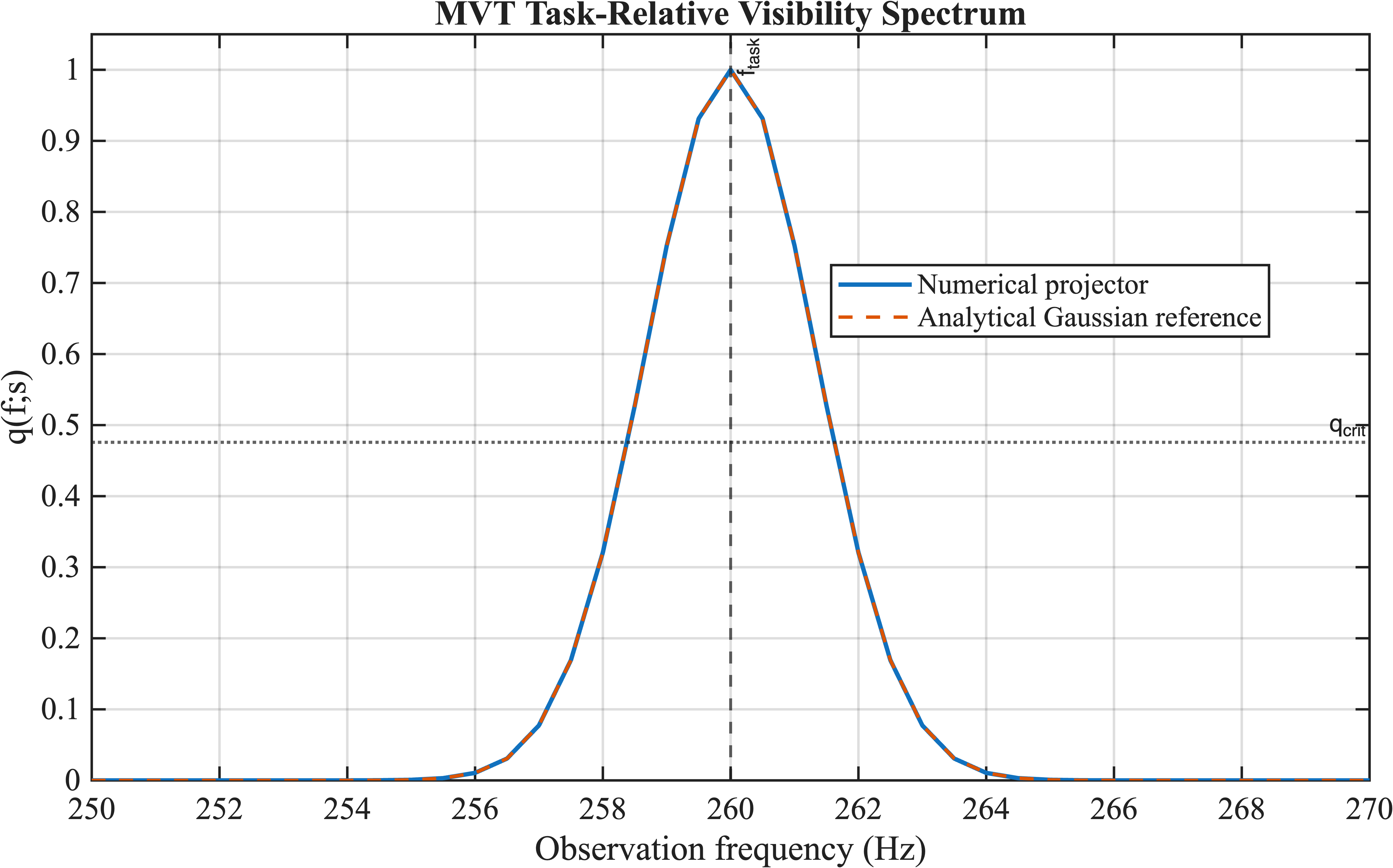}}
\hfill
\subfloat[Detectability spectrum and decision-visible region.]{%
\includegraphics[width=0.49\textwidth]
{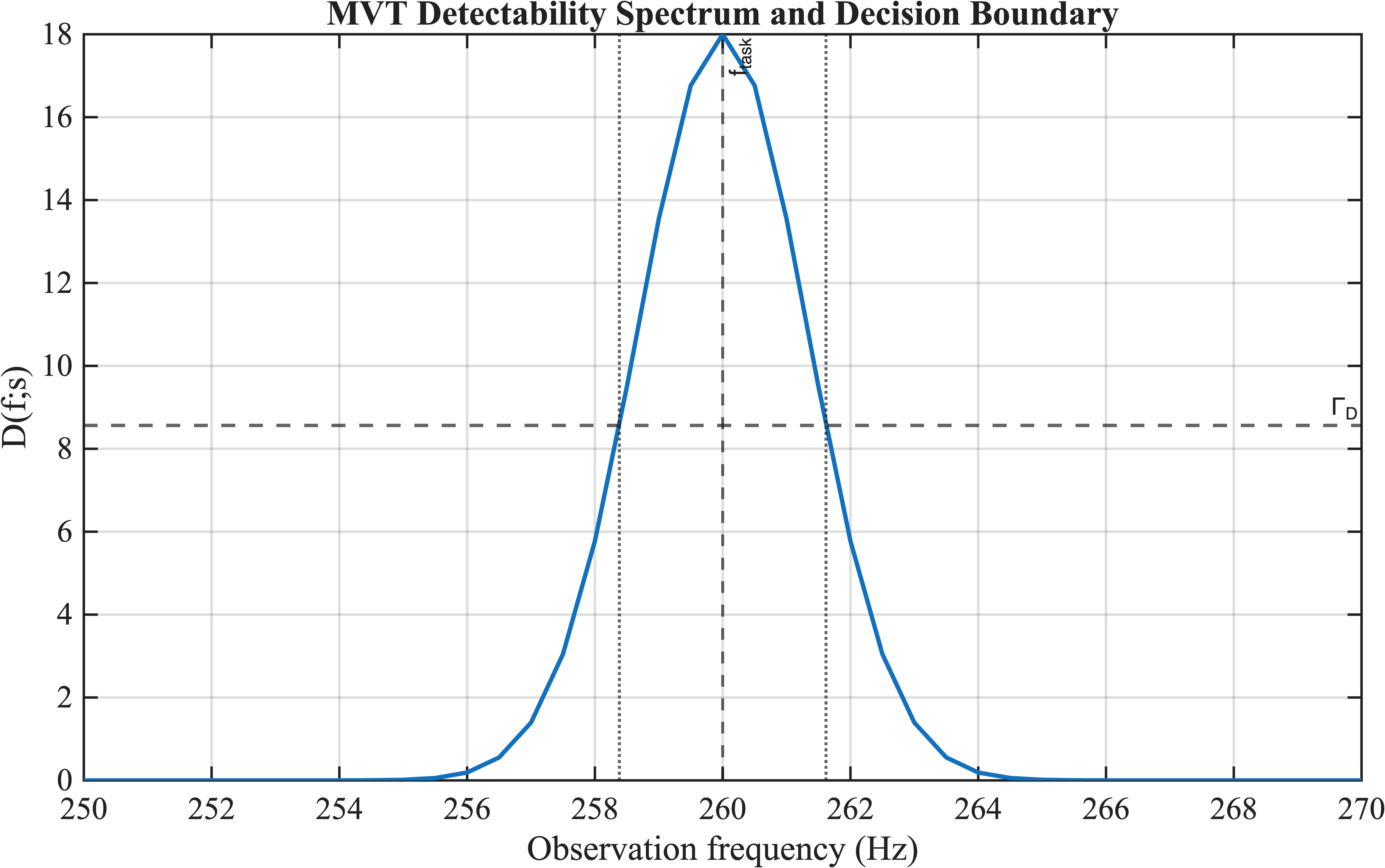}}
\caption{
Frequency-indexed MVT analysis for the prescribed 260-Hz task.
(a) Numerical geometric visibility agrees with the analytical Gaussian
reference in \eqref{eq:gaussian_visibility_reference} and reaches
$q_{\max}=1$ at $f_q^\star=260$~Hz.
(b) Under the stated white-input-noise model, the corresponding statistical
separation reaches $D_{\max}=18$ and exceeds
$\Gamma_D=8.563847350668$ only within the decision-visible neighborhood.
The discrete and continuous boundaries are given by
\eqref{eq:grid_interval} and \eqref{eq:continuous_interval}, respectively.
}
\label{fig:spectra}
\end{figure*}

The distinction between geometry and statistics is also evident from Theorem~\ref{thm:geometric_statistical}. For fixed task and observation geometry, increasing the pre-observation white-noise variance reduces $D$ through the inverse dependence on $\sigma^2$ but does not alter $q$. By contrast, temporal or frequency mismatch changes alignment between the task and the observation subspace and therefore modifies $q$ itself. These statements apply to the specified observation and uncertainty models and are not asserted as universal invariance properties.

\subsection{Robustness Under Partial Observation}

The controlled signal study also instantiates the erasure analysis of Section~\ref{sec:robust_collective}. For the minimal two-row observer and the four-row phase-redundant observer, exact enumeration gives
\begin{equation}
G_r^{(2)}
=
(1,0,0),
\qquad
G_r^{(4)}
=
(1,1,1,0,0),
\label{eq:signal_erasure_result}
\end{equation}
in agreement with \eqref{eq:erasureseq}.

As shown in Fig.~\ref{fig:erasure}, the nominal visibility of the two systems is identical before erasure, but their worst-case behavior under partial loss differs. The result confirms that measurement redundancy can preserve alternative access paths to the prescribed task subspace even though it cannot recreate information that was absent from the upstream observation.

\subsection{Collective Visibility and Joint Decision Sufficiency}

To evaluate collective visibility, four local observation branches are centered at
\begin{equation}
\mathcal F_{\rm cand}
=
\{257,\;258,\;262,\;263\}
~\mathrm{Hz}.
\label{eq:collective_frequencies}
\end{equation}
The 258-Hz and 262-Hz branches are individually below the prescribed decision threshold:
\begin{equation}
D_{258}
=
D_{262}
=
5.774163115783
<
\Gamma_D.
\label{eq:local_D}
\end{equation}
Their genuinely joint observation, however, gives
\begin{equation}
D_{258,262}
=
10.470833520729
>
\Gamma_D.
\label{eq:joint_D}
\end{equation}
Equations~\eqref{eq:local_D} and \eqref{eq:joint_D} therefore satisfy the purely collective decision conditions in
\eqref{eq:individual_subthreshold} and
\eqref{eq:collective_threshold}.

The corresponding single-branch geometric visibilities are
\begin{equation}
q_{258}
=
q_{262}
=
0.320786839766,
\label{eq:collective_q_single}
\end{equation}
whereas the joint accessible subspace gives
\begin{equation}
q_{258,262}
=
0.581712973374.
\label{eq:collective_q_joint}
\end{equation}
The conditional geometric innovation of the 262-Hz branch when added to the 258-Hz branch is therefore
\begin{equation}
\delta q_{262|\{258\}}
=
0.260926133608.
\label{eq:collective_innovation_result}
\end{equation}
Equation~\eqref{eq:collective_innovation_result} quantifies the additional task-accessible component contributed by the second branch relative to the already available 258-Hz subspace.

Figure~\ref{fig:collective_results}(a) summarizes the worst, mean, and best subset visibilities defined in \eqref{eq:collective_envelopes}. Figure~\ref{fig:collective_results}(b) shows the statistical consequence for the selected 258-Hz and 262-Hz pair.

\begin{figure*}[t]
\centering
\subfloat[Collective best--mean--worst visibility envelopes.]{%
\includegraphics[width=0.49\textwidth]
{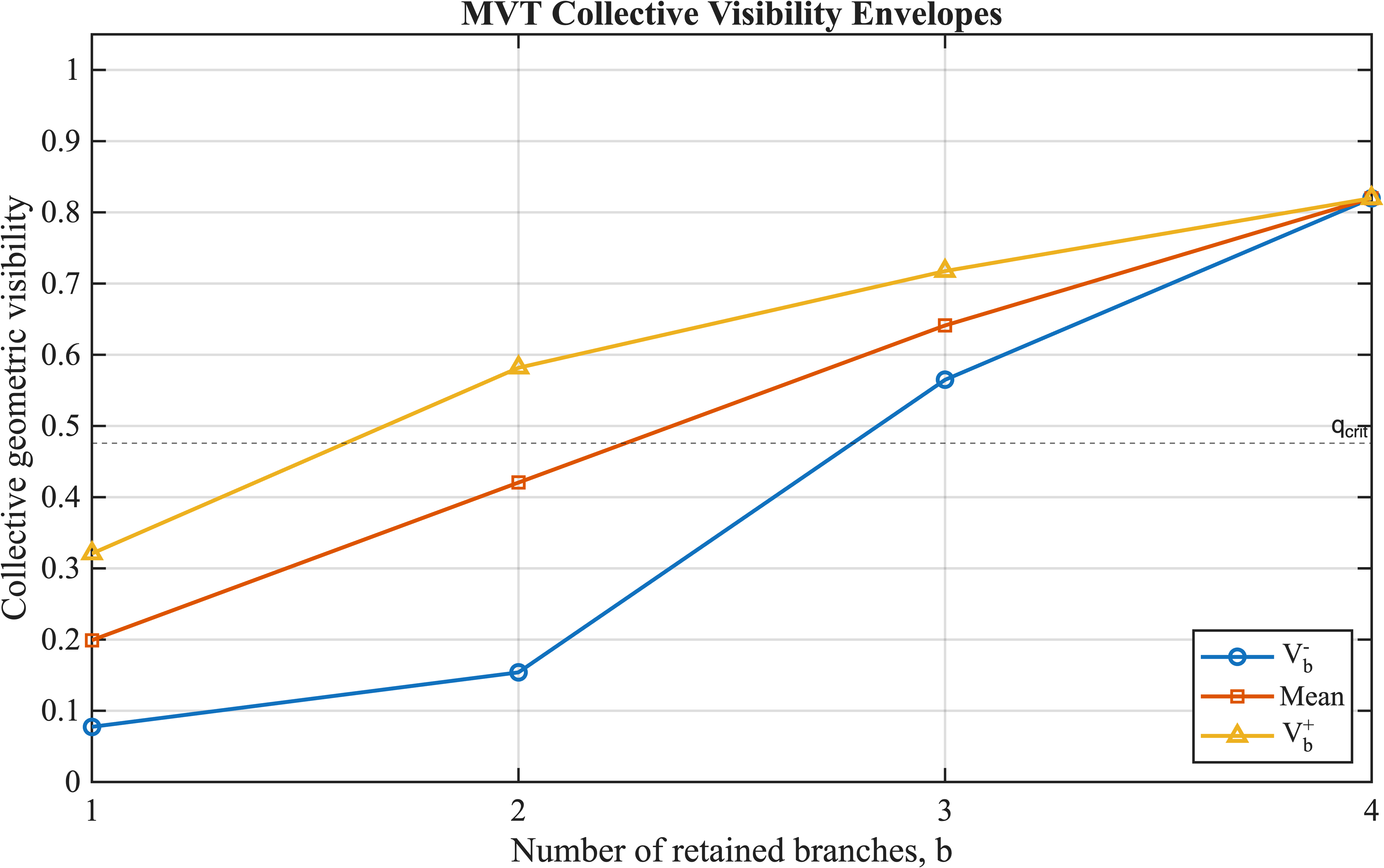}}
\hfill
\subfloat[Purely collective decision visibility.]{%
\includegraphics[width=0.49\textwidth]
{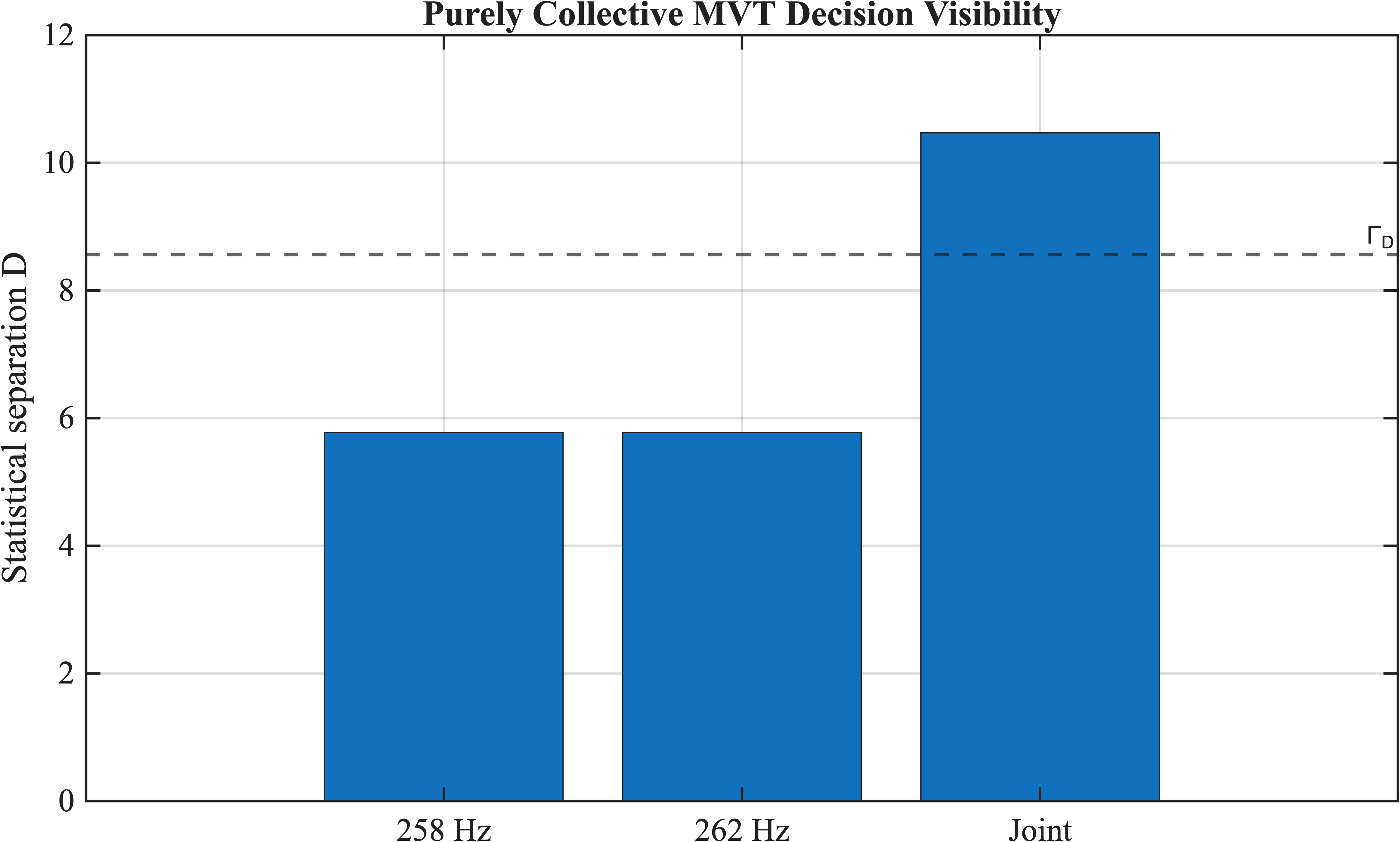}}
\caption{
Collective visibility for the four candidate observation branches in
\eqref{eq:collective_frequencies}.
(a) The envelopes $V_b^-$, $\overline V_b$, and $V_b^+$ defined in
\eqref{eq:collective_envelopes} distinguish the number of retained branches
from the identity of the retained pattern; the horizontal reference is the
decision-equivalent geometric threshold $q_{\rm crit}$.
(b) The 258-Hz and 262-Hz branches are individually below
$\Gamma_D$, whereas their joint observation exceeds the required threshold,
as quantified by \eqref{eq:local_D} and \eqref{eq:joint_D}.
}
\label{fig:collective_results}
\end{figure*}

The collective result does not imply that either individual branch contains no task information. Both have positive geometric visibility, but neither provides sufficient statistical separation for the specified operating point. Joint access enlarges the task-relevant accessible subspace, increases geometric visibility in accordance with Proposition~\ref{prop:collective_monotonicity}, and raises the statistical separation above the decision threshold.

The controlled study therefore establishes four distinct observations under known conditions. First, complete-signal energy dominance and prescribed-task visibility can occur at different frequencies within the same observation family. Second, geometric visibility and statistical decision sufficiency remain separate quantities. Third, redundant observation can preserve task accessibility under erasure without restoring information lost upstream. Fourth, complementary individually subthreshold observations can become jointly decision-visible. These conclusions are specific to the prescribed benchmark and support the internal interpretation of MVT; they do not constitute a claim of universal superiority over FFT, STFT, CWT, HHT, or other signal representations.

\section{Discussion and Relation to Established Signal Representations}
\label{sec:discussion}

The controlled results clarify the role of MVT relative to established
signal-processing methods. MVT is not intended to replace Fourier, STFT,
wavelet, or Hilbert--Huang analysis. These methods characterize different
properties of a signal: Fourier analysis describes global spectral content,
STFT provides localized time--frequency information, wavelet methods provide
multiresolution representations, and HHT constructs adaptive empirical modes
with instantaneous-frequency descriptions
\cite{cohen1995,mallat2009,huang1998}. In contrast, MVT addresses a
task-relative question: whether a prescribed direction or task model remains
accessible through a specified observation or representation architecture
and, after an uncertainty model is introduced, whether the retained evidence
is sufficient for a required decision.

The distinction can be summarized as
\begin{equation}
\begin{aligned}
\text{signal presence}
&\neq
\text{representation prominence}
\\
&\neq
\text{geometric accessibility}
\\
&\neq
\text{statistical detectability}
\\
&\neq
\text{decision sufficiency}.
\end{aligned}
\label{eq:interpretation_hierarchy}
\end{equation}
Equation~\eqref{eq:interpretation_hierarchy} expresses the central interpretive
separation of the framework. A component may be physically present but poorly
represented, represented strongly but irrelevant to the prescribed task,
geometrically accessible but statistically weak, or statistically separated
yet still below the operating requirement imposed by $\Gamma_D$.

\subsection{What MVT Adds Beyond Representation Magnitude}

For a compatible linear front end, an established signal representation can
be incorporated into the effective observation operator. The resulting
analysis chain is

\begin{equation}
\begin{aligned}
\text{representation}
&\longrightarrow
A_{\mathcal A}
\longrightarrow
q_{\mathcal A}(s)
\\
&\longrightarrow
D_{\mathcal A}(s)
\longrightarrow
\bigl[D_{\mathcal A}(s)\ge\Gamma_D\bigr].
\end{aligned}
\label{eq:representation_chain}
\end{equation}

Equation~\eqref{eq:representation_chain} shows that the representation and MVT
answer different questions. The representation determines how the data are
observed or encoded; MVT evaluates the geometry of the prescribed task relative
to the resulting accessible subspace and then, under a stated uncertainty
model, connects that geometry to statistical and decision-level quantities.

This distinction is illustrated by the controlled study in
Section~\ref{sec:controlled_study}. Within the same Gaussian-localized
observation family, the complete-signal projected energy is maximized at
380~Hz, while the prescribed-task visibility is maximized at 260~Hz. The
result in \eqref{eq:keyresult} therefore does not claim that MVT estimates
frequency more accurately than FFT, STFT, CWT, or HHT. Rather, it demonstrates
that the frequency associated with maximum signal prominence need not coincide
with the frequency at which a specified task direction is most accessible.

The characteristic frequencies reported in Table~\ref{tab:methods} should
consequently not be interpreted as competing estimates of a common latent
parameter. They are maxima of different representations or task-dependent
quantities. The comparison is intended to expose the distinction between
representation prominence and task-relative accessibility, not to establish
a universal ranking among signal-processing methods.

\subsection{Relation to Task-Based Observation and Subspace Methods}

The task-relative viewpoint has connections with task-based acquisition,
matched-subspace detection, projection-based fault diagnosis, functional
observability, frame theory, and active-subspace methods
\cite{shlezinger2019,bernardo2023,scharf1994,
DingLiLiu2026ProjectionFaultDiagnosis,fernando2010,darouach2025,
christensen2016,constantine2014,constantine2017}. These areas share elements
of sensing, subspace geometry, redundancy, or task dependence, but they address
different primary objects and design questions.

Task-based acquisition methods typically optimize sensing or quantization for
a prescribed downstream objective. Functional observability concerns whether
selected functions of a dynamical state can be reconstructed from system
outputs. Matched-subspace detection evaluates statistical hypotheses associated
with structured signal subspaces. Frame-theoretic methods characterize
redundant representations and their robustness to coefficient loss.
Active-subspace methods identify influential parameter directions from model
sensitivity.

Projection-based fault diagnosis provides a particularly relevant comparison
because it also makes explicit use of projection and subspace geometry.
Ding, Li, and Liu represent dynamic-system behavior through Hilbert-space
subspaces, construct projection-based residuals relative to system behavior,
and use subspace-distance and gap-metric concepts for fault detection and
isolation \cite{DingLiLiu2026ProjectionFaultDiagnosis}. That framework solves
a system-discrimination problem: whether measured behavior is sufficiently
separated from nominal or fault-associated system subspaces. MVT addresses a
different object. It begins with an independently prescribed task direction or
task subspace and evaluates how much of that task remains accessible through a
composed observation architecture. Projection is therefore a shared
mathematical tool rather than the claimed novelty of MVT.

MVT extends this task-relative viewpoint by explicitly separating questions
that are otherwise easy to conflate: whether the prescribed task survives,
how much of it lies in the observation-accessible row space, where loss occurs
along a multistage observation chain, how accessibility changes with scale,
how robust it is to partial row loss, what complementary observations add,
and whether the resulting evidence is statistically sufficient for a stated
operating point. These distinctions lead to the task-specific constructions
developed in this paper, including conditional stagewise retention,
scale-indexed visibility and stability bounds, worst-case erasure visibility,
collective visibility, and the explicit separation of geometric visibility
from statistical and decision sufficiency. The novelty claimed here therefore
lies in this unified task-relative operator-geometric organization and its
associated constructions, not in orthogonal projection, subspace distance,
gap metrics, principal angles, frame theory, or Gaussian detection
individually.

This organization is built around the observation-accessible row space
\begin{equation}
\mathcal O_{\mathcal A}
=
\mathcal R(A_{\mathcal A}^{H}),
\label{eq:discussion_accessible_space}
\end{equation}
but the row space alone is not treated as a complete measure of system
quality. Its relevance is always evaluated relative to a prescribed task and,
where required, to an explicit uncertainty and decision model. This
task-relative restriction is what permits the same observation architecture
to be examined separately in terms of preservation, robustness, multiscale
variation, complementarity, and detectability.

\subsection{Task-Relative Representation Selection}

For a compatible family of linear representations indexed by $\rho$, define
\begin{equation}
A_{\rho,\beta}
=
R_E
W_\rho
H_\beta
P_\Omega .
\label{eq:representation_family}
\end{equation}
Equation~\eqref{eq:representation_family} keeps acquisition,
scale-dependent observation, representation, and retained access explicit
while allowing the representation family to vary.

A representation may then be selected according to geometric visibility,
\begin{equation}
\rho_q^\star
\in
\operatorname*{arg\,max}_{\rho}
q_{A_{\rho,\beta}}(s),
\label{eq:rho_q}
\end{equation}
or according to statistical separation,
\begin{equation}
\rho_D^\star
\in
\operatorname*{arg\,max}_{\rho}
D_{A_{\rho,\beta}}(s).
\label{eq:rho_D}
\end{equation}
Equations~\eqref{eq:rho_q} and \eqref{eq:rho_D} define two different
task-relative design criteria. The first depends only on accessible geometry,
whereas the second also depends on the uncertainty model and noise placement.
The corresponding optimizers therefore need not coincide.

More generally, representation selection can be formulated over both
representation and scale:
\begin{equation}
(\rho^\star,\beta^\star)
\in
\operatorname*{arg\,max}_{\rho,\beta}
\mathcal J_{\rm MVT}(\rho,\beta),
\label{eq:joint_representation_selection}
\end{equation}
where $\mathcal J_{\rm MVT}$ may be chosen as $q$, $q_{\min}$, $D$, a
robustness certificate, or another explicitly stated task-relative criterion.
Equation~\eqref{eq:joint_representation_selection} is a design formulation
rather than a claim that one representation family is universally preferable.

Accordingly, MVT does not imply an intrinsic ordering such as
\[
\text{STFT}>\text{wavelet}
\quad\text{or}\quad
\text{wavelet}>\text{STFT}.
\]
Any preference is conditional on the prescribed task, observation
architecture, scale family, retained access, and statistical model.

\subsection{Treatment of Adaptive and Nonlinear Representations}

Linear Fourier, STFT, filter-bank, and fixed wavelet operators can enter the
MVT architecture directly through $W_\rho$ or $H_\beta$. HHT requires a
different interpretation because empirical mode decomposition is data adaptive
and generally nonlinear \cite{huang1998}. In the present study, HHT is
therefore treated as an external comparison representation rather than as a
linear MVT operator.

Embedding HHT or another nonlinear adaptive representation directly into MVT
would require an explicit nonlinear extension. A local Jacobian could provide
a first-order linear approximation around a specified operating point, but
such a construction would describe only local behavior and would not establish
global nonlinear visibility. No such extension is claimed in the present
paper.

\subsection{Local Time--Scale Visibility}

A time--scale extension requires both a local task definition and an explicitly
localized observation operator. Let $L_\tau$ denote a prescribed localization
operator centered at local coordinate $\tau$. For a global task direction $s$,
define the localized task component, when nonzero, as
\begin{equation}
s_\tau
=
\frac{L_\tau s}
{\|L_\tau s\|}.
\label{eq:local_task}
\end{equation}
Equation~\eqref{eq:local_task} specifies what task is being evaluated locally
rather than implicitly interpreting a transform coefficient as task
information.

A compatible local observation family can then be written as
\begin{equation}
A_{\tau,\beta}
=
R_E
W_\rho
H_\beta
L_\tau
P_\Omega .
\label{eq:local_operator}
\end{equation}
Equation~\eqref{eq:local_operator} makes the local coordinate $\tau$ part of
the observation architecture through the explicit localization operator
$L_\tau$.

The corresponding local time--scale visibility is
\begin{equation}
q(\tau,\beta;s)
=
\frac{
\left\|
P_{\mathcal R(A_{\tau,\beta}^{H})}
s_\tau
\right\|^2
}{
\|s_\tau\|^2
}.
\label{eq:timescale_q}
\end{equation}
Since $s_\tau$ is normalized in \eqref{eq:local_task}, the denominator in
\eqref{eq:timescale_q} equals one when the local task is defined.
Equation~\eqref{eq:timescale_q} is therefore a task-relative projection field
over a family of localized observation subspaces.

This construction is not an STFT spectrogram or wavelet scalogram. A
spectrogram or scalogram displays transform magnitude or energy as a function
of time and frequency or scale. In contrast, \eqref{eq:timescale_q} evaluates
the alignment of a prescribed local task with the row space of an explicitly
defined local observation operator.

\subsection{Controlled Two-Parameter Illustration}

The paper-level time--scale figure uses a deliberately simple analytical
operator family to illustrate the geometry independently of a particular
signal transform. Let
\begin{equation}
A_{\beta,\tau}
=
\begin{bmatrix}
\cos\theta(\beta,\tau)
&
\sin\theta(\beta,\tau)
\end{bmatrix},
\label{eq:timescale_operator_a}
\end{equation}
with
\begin{equation}
\theta(\beta,\tau)
=
\beta
+
0.70\pi\tau,
\qquad
\beta\in[0,\pi],
\quad
\tau\in[-1,1].
\label{eq:timescale_theta}
\end{equation}
Equations~\eqref{eq:timescale_operator_a} and
\eqref{eq:timescale_theta} define a controlled two-parameter path of
one-dimensional accessible subspaces. In this illustrative construction,
$\tau$ is a local-coordinate parameter of the operator family; it should not
be interpreted as a measured physical time axis unless the operator is
explicitly linked to a localization operator such as $L_\tau$ in
\eqref{eq:local_operator}.

For the prescribed direction
\begin{equation}
s=e_1,
\label{eq:timescale_task}
\end{equation}
the geometric visibility is exactly
\begin{equation}
q(\beta,\tau;s)
=
\cos^2
\theta(\beta,\tau).
\label{eq:timescale_closed}
\end{equation}
Equation~\eqref{eq:timescale_closed} follows directly from projection onto the
one-dimensional row space of \eqref{eq:timescale_operator_a}.

For the controlled statistical illustration, define
\begin{equation}
D(\beta,\tau;s)
=
5q(\beta,\tau;s),
\qquad
\Gamma_D
=
2.5.
\label{eq:timescale_D}
\end{equation}
Equation~\eqref{eq:timescale_D} is an illustrative scaling chosen to
demonstrate how a geometric visibility field induces a corresponding
decision-visible region; it is not presented as a universal statistical
model.

The decision-visible set for this example is
\begin{equation}
\mathcal V_D
=
\left\{
(\beta,\tau):
D(\beta,\tau;s)
\ge
\Gamma_D
\right\}.
\label{eq:timescale_visible_set}
\end{equation}
Equation~\eqref{eq:timescale_visible_set} converts the two-parameter geometric
field into a decision region under the stated scaling.

Figure~\ref{fig:timescale} shows the geometric and decision-level forms of
this controlled construction. Figure~\ref{fig:timescale}(a) displays the
task-relative projection field, while Fig.~\ref{fig:timescale}(b) shows the
corresponding decision margin and the boundary $D=\Gamma_D$.

\begin{figure*}[t]
\centering
\subfloat[Task-relative two-parameter visibility field.]{%
\includegraphics[width=0.49\textwidth]
{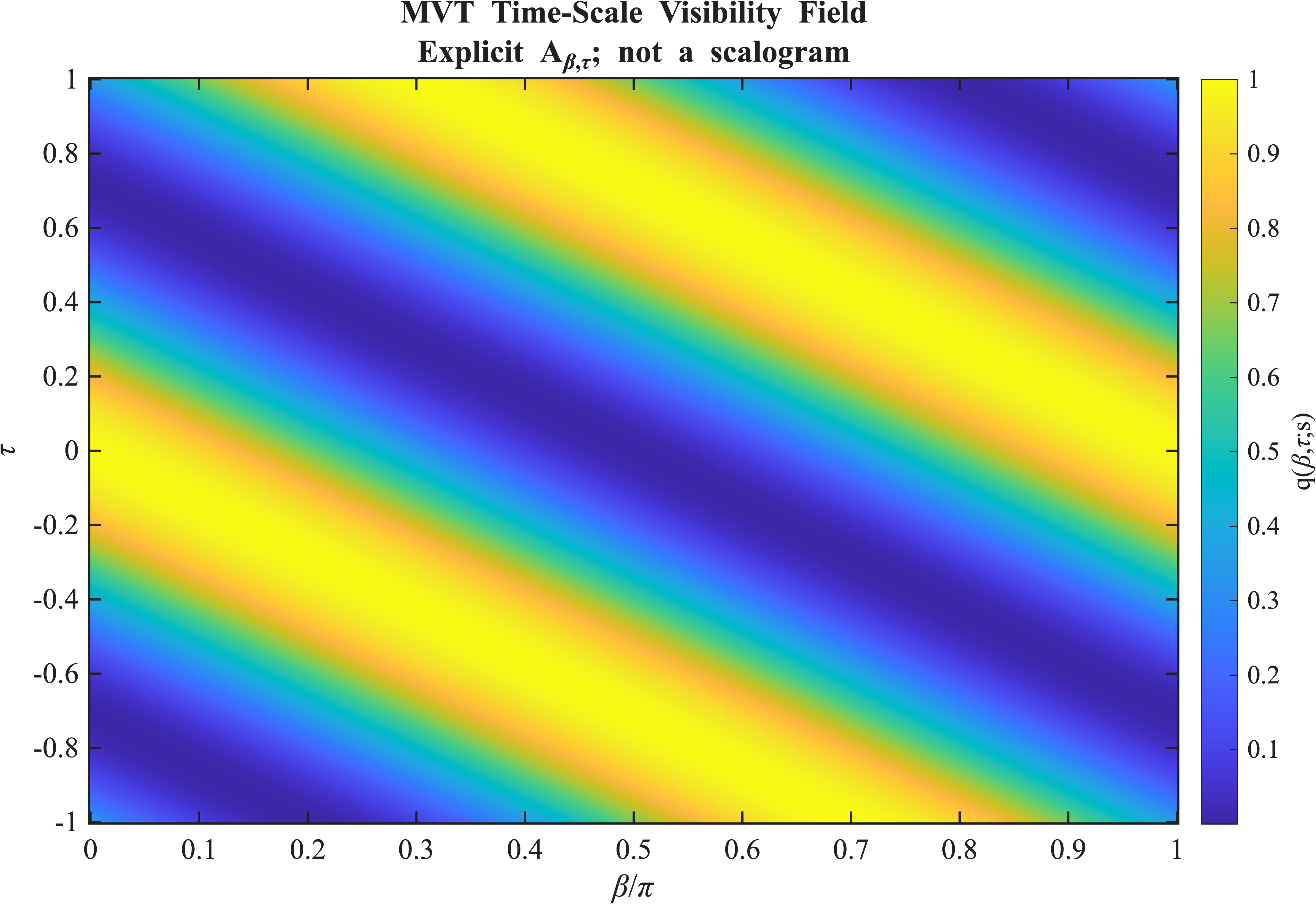}}
\hfill
\subfloat[Decision margin and decision-visible boundary.]{%
\includegraphics[width=0.49\textwidth]
{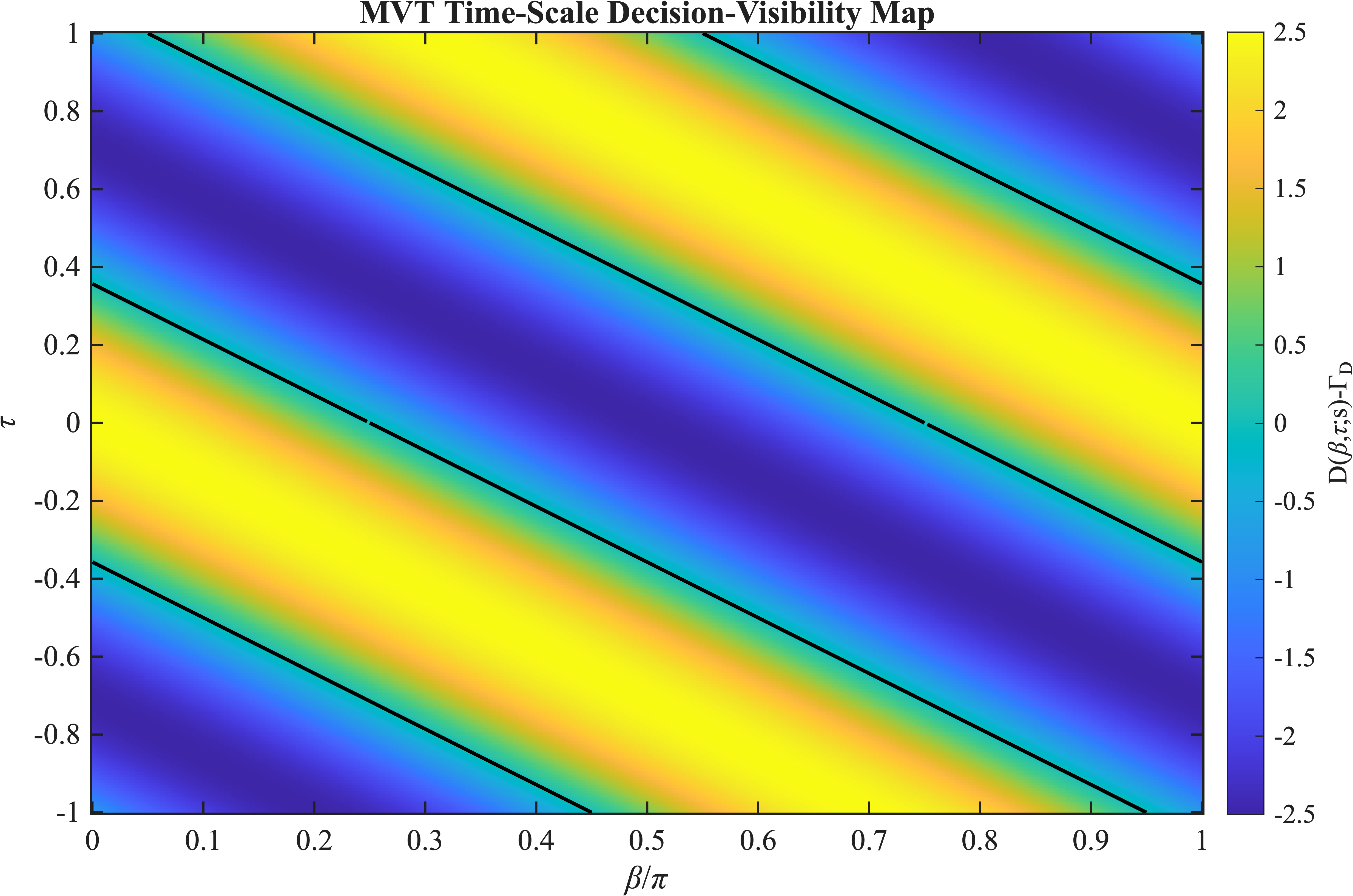}}
\caption{
Controlled two-parameter MVT illustration based on the explicit observation
family in \eqref{eq:timescale_operator_a}--\eqref{eq:timescale_theta}.
(a) The field $q(\beta,\tau;s)$ is the task-relative projection quantity in
\eqref{eq:timescale_closed}, not a transform-energy scalogram.
(b) Decision margin generated by \eqref{eq:timescale_D}; the zero contours
identify the boundary $D(\beta,\tau;s)=\Gamma_D$ of the set in
\eqref{eq:timescale_visible_set}.
The example illustrates the geometry of a two-parameter observation family
and does not claim a general time--frequency representation theorem.
}
\label{fig:timescale}
\end{figure*}

\subsection{Interpretation of Robust and Collective Visibility}

The erasure and collective results address two different consequences of
observation multiplicity. Redundancy improves robustness when alternative
retained rows preserve the same prescribed task subspace, as illustrated by
the four-row observer in Fig.~\ref{fig:erasure}. This does not create new task
information; it preserves alternative access to information that was already
present in the nominal observation geometry.

Collective visibility addresses a different effect. Two observation branches
may each intersect the task direction only partially, while their joint span
contains a substantially larger task component.
Proposition~\ref{prop:collective_monotonicity} guarantees that adding a
genuinely available branch cannot reduce geometric visibility, but the amount
of improvement remains task dependent. The controlled result in
\eqref{eq:local_D}--\eqref{eq:joint_D} shows that this geometric
complementarity can also change a statistical conclusion: two individually
subthreshold observations become jointly sufficient for the prescribed
operating point.

These two mechanisms should therefore not be conflated. Erasure robustness
concerns preservation under loss, whereas collective visibility concerns gain
from complementary access.

\subsection{Scope of the Present Evidence}

The analytical results establish properties of the finite-dimensional linear
MVT framework. The numerical audits establish finite-precision consistency of
the corresponding implementation, and the synthetic benchmark demonstrates
the interpretation of the framework under controlled conditions. None of
these results establishes universal empirical superiority over existing signal
representations.

The present evidence supports the more limited conclusion that MVT provides
an additional task-relative analysis layer when the engineering objective is
defined by prescribed information rather than by maximum representation
magnitude alone. Whether this additional layer is practically useful in a
specific engineering application must be established using independently
defined task signatures, observation families, uncertainty models, and
measured data.

\section{Scope, Limitations, and Reproducibility}

The established core in this paper is finite dimensional and linear at the observation level. The task direction or task subspace is prescribed rather than discovered automatically. The statistical formulas are model dependent: \eqref{eq:gamma} uses an equal-covariance Gaussian mean-shift model, while \eqref{eq:Dwhite} additionally requires pre-observation isotropic white noise that passes through the same deterministic operator as the task displacement.

The controlled benchmark has known ground truth and tests interpretation and implementation under specified conditions. It does not establish broad empirical superiority or measured-data performance. Application-specific claims require measured data for which the task signature, observation family, uncertainty model, erasure mechanism, and decision criterion are defined before analysis.

All numerical results reported here were obtained in MATLAB R2025b Update 1 using a nominal audit tolerance of $10^{-10}$. The paper-level implementation stores numerical results and figure data from a common run; residuals near machine precision are reported explicitly. The 12/12 result in Table~\ref{tab:audit} certifies consistency of the implemented calculations with the tested analytical relations only.

Potential nonlinear extensions can be motivated locally through Jacobian geometry, but a local linearization does not establish global nonlinear visibility or detectability. Learned or active-subspace representations may likewise provide candidate task directions or observation families, but they do not remove the requirement to state the ambient space, the task model, and the accessible operator geometry explicitly \cite{constantine2014,constantine2017}.

\section{Conclusion}

This paper formulated Multiscale Visibility Theory as a task-relative operator-geometric framework for determining what prescribed information survives an observation architecture, how much of that information remains geometrically accessible, how robust that access is to partial availability and scale change, what genuinely complementary observations add, and whether the surviving evidence is statistically sufficient for a stated decision criterion.

The principal theoretical contribution is the organization of these questions within one architecture while preserving their mathematical distinctions. The geometric core is the row-space projection in \eqref{eq:q}; the task-model extension is given by \eqref{eq:G}--\eqref{eq:qmin}; stagewise loss is attributed exactly by \eqref{eq:factor}; statistical separation is connected to geometry by \eqref{eq:Dwhite}; scale variation is controlled by the projector bound in \eqref{eq:stability}; worst-case partial access is quantified by \eqref{eq:Gr}; and collective complementarity is described by \eqref{eq:collective_envelopes} and \eqref{eq:innovation}. The novelty is this task-relative synthesis, not the individual projection, frame, perturbation, or Gaussian-detection tools on which it is built.

The numerical findings support the internal consistency of that formulation. The stagewise direct and factorized computations both yield $0.388962986023$. In the controlled signal benchmark, the same frequency-indexed MVT observation family produces $f_q^\star=260$ Hz for prescribed-task visibility but $f_E^\star=380$ Hz for complete-signal observation energy. Thus, the result in \eqref{eq:keyresult} demonstrates within a common observation family that energy dominance and task visibility are not interchangeable. At the nominal noise level, $D_{\max}=18$ exceeds $\Gamma_D=8.563847350668$, and the numerical decision-visible interval 258.5--261.5 Hz agrees with the continuous Gaussian reference 258.383423--261.616577 Hz.

The robustness and collective experiments provide two further findings. The four-row phase-redundant observer preserves full worst-case task-model visibility through two row erasures, whereas the minimal two-row observer fails after one. Separately, the 258-Hz and 262-Hz branches are each below the required decision threshold with $D_1=D_2=5.774163115783$, but their joint observation reaches $D_{12}=10.470833520729$ with conditional geometric innovation $0.260926133608$. These results show, respectively, that redundancy can preserve alternative access paths without creating new upstream information and that complementary observations can become jointly decision-visible even when each is individually insufficient.

All 12 registered numerical checks passed at the stated $10^{-10}$ tolerance. These results establish reproducibility and finite-precision consistency for the controlled calculations, not universal empirical superiority. MVT is therefore positioned as a task-relative analysis and design layer that can operate with established representations while keeping representation magnitude, geometric accessibility, robustness, statistical separation, and decision sufficiency distinct. Measured engineering data, nonlinear observation maps, adaptive representations, and broader stochastic settings remain subjects for subsequent validation and extension.

\bibliographystyle{IEEEtran}
\bibliography{references}

@book{daubechies1992,
  author    = {Ingrid Daubechies},
  title     = {Ten Lectures on Wavelets},
  series    = {CBMS-NSF Regional Conference Series in Applied Mathematics},
  volume    = {61},
  publisher = {SIAM},
  address   = {Philadelphia, PA, USA},
  year      = {1992},
  doi       = {10.1137/1.9781611970104}
}

@book{mallat2009,
  author    = {St{\'e}phane Mallat},
  title     = {A Wavelet Tour of Signal Processing: The Sparse Way},
  edition   = {3},
  publisher = {Academic Press},
  address   = {Amsterdam, The Netherlands},
  year      = {2009},
  doi       = {10.1016/B978-0-12-374370-1.X0001-8}
}

@book{lindeberg1994,
  author    = {Tony Lindeberg},
  title     = {Scale-Space Theory in Computer Vision},
  publisher = {Springer},
  address   = {New York, NY, USA},
  year      = {1994},
  doi       = {10.1007/978-1-4757-6465-9}
}

@book{cohen1995,
  author    = {Leon Cohen},
  title     = {Time-Frequency Analysis},
  publisher = {Prentice-Hall PTR},
  address   = {Englewood Cliffs, NJ, USA},
  year      = {1995}
}

@article{huang1998,
  author  = {Norden E. Huang and Zheng Shen and Steven R. Long and Manli C. Wu and Hsing H. Shih and Quanan Zheng and Nai-Chyuan Yen and Chi Chao Tung and Henry H. Liu},
  title   = {The empirical mode decomposition and the Hilbert spectrum for nonlinear and non-stationary time series analysis},
  journal = {Proceedings of the Royal Society of London. Series A},
  volume  = {454},
  number  = {1971},
  pages   = {903--995},
  year    = {1998},
  doi     = {10.1098/rspa.1998.0193}
}

@book{horn2013matrix,
  author    = {Roger A. Horn and Charles R. Johnson},
  title     = {Matrix Analysis},
  edition   = {2},
  publisher = {Cambridge University Press},
  address   = {Cambridge, U.K.},
  year      = {2013}
}

@article{bjorckgolub1973,
  author  = {{\AA}ke Bj{\"o}rck and Gene H. Golub},
  title   = {Numerical methods for computing angles between linear subspaces},
  journal = {Mathematics of Computation},
  volume  = {27},
  number  = {123},
  pages   = {579--594},
  year    = {1973},
  doi     = {10.1090/S0025-5718-1973-0348991-3}
}

@article{daviskahan1970,
  author  = {Chandler Davis and William M. Kahan},
  title   = {The rotation of eigenvectors by a perturbation. III},
  journal = {SIAM Journal on Numerical Analysis},
  volume  = {7},
  number  = {1},
  pages   = {1--46},
  year    = {1970},
  doi     = {10.1137/0707001}
}

@article{edelman1998,
  author  = {Alan Edelman and Tom{\'a}s A. Arias and Steven T. Smith},
  title   = {The geometry of algorithms with orthogonality constraints},
  journal = {SIAM Journal on Matrix Analysis and Applications},
  volume  = {20},
  number  = {2},
  pages   = {303--353},
  year    = {1998},
  doi     = {10.1137/S0895479895290954}
}

@book{stewart1990,
  author    = {G. W. Stewart and Ji-Guang Sun},
  title     = {Matrix Perturbation Theory},
  publisher = {Academic Press},
  address   = {Boston, MA, USA},
  year      = {1990}
}

@book{christensen2016,
  author    = {Ole Christensen},
  title     = {An Introduction to Frames and Riesz Bases},
  edition   = {2},
  publisher = {Birkh{\"a}user},
  address   = {Cham, Switzerland},
  year      = {2016},
  doi       = {10.1007/978-3-319-25613-9}
}

@book{casazza2013,
  editor    = {Peter G. Casazza and Gitta Kutyniok},
  title     = {Finite Frames: Theory and Applications},
  publisher = {Birkh{\"a}user},
  address   = {Boston, MA, USA},
  year      = {2013},
  doi       = {10.1007/978-0-8176-8373-3}
}

@article{bodmann2005,
  author  = {Bernhard G. Bodmann and Vern I. Paulsen},
  title   = {Frames, graphs and erasures},
  journal = {Linear Algebra and its Applications},
  volume  = {404},
  pages   = {118--146},
  year    = {2005},
  doi     = {10.1016/j.laa.2005.02.016}
}

@article{holmes2004optimal,
  author  = {Richard B. Holmes and Vern I. Paulsen},
  title   = {Optimal frames for erasures},
  journal = {Linear Algebra and its Applications},
  volume  = {377},
  pages   = {31--51},
  year    = {2004}
}

@article{fickus2012numerically,
  author  = {Matthew Fickus and Dustin G. Mixon},
  title   = {Numerically erasure-robust frames},
  journal = {Linear Algebra and its Applications},
  volume  = {437},
  number  = {6},
  pages   = {1394--1407},
  year    = {2012}
}

@article{fickus2014group,
  author  = {Matthew Fickus and Dustin G. Mixon and John Jasper and Jesse Peterson},
  title   = {Group-theoretic constructions of erasure-robust frames},
  journal = {Applied and Computational Harmonic Analysis},
  volume  = {37},
  number  = {3},
  pages   = {511--530},
  year    = {2014}
}

@article{chen2026,
  author  = {R. Chen and R. Li and Y. Guo},
  title   = {Modeling and mining for multiscale visibility data of landscapes},
  journal = {npj Heritage Science},
  volume  = {14},
  pages   = {324},
  year    = {2026},
  doi     = {10.1038/s40494-026-02554-z}
}

@article{candes2005decoding,
  author  = {Emmanuel J. Cand{\`e}s and Terence Tao},
  title   = {Decoding by linear programming},
  journal = {IEEE Transactions on Information Theory},
  volume  = {51},
  number  = {12},
  pages   = {4203--4215},
  month   = dec,
  year    = {2005}
}

@book{foucart2013mathematical,
  author    = {Simon Foucart and Holger Rauhut},
  title     = {A Mathematical Introduction to Compressive Sensing},
  publisher = {Birkh{\"a}user},
  address   = {New York, NY, USA},
  year      = {2013}
}

@article{shlezinger2019,
  author  = {Nir Shlezinger and Yonina C. Eldar and Miguel R. D. Rodrigues},
  title   = {Hardware-limited task-based quantization},
  journal = {IEEE Transactions on Signal Processing},
  volume  = {67},
  number  = {20},
  pages   = {5223--5238},
  year    = {2019},
  doi     = {10.1109/TSP.2019.2935864}
}

@article{bernardo2023,
  author  = {N. I. Bernardo and J. Zhu and Y. C. Eldar and J. Evans},
  title   = {Design and analysis of hardware-limited non-uniform task-based quantizers},
  journal = {IEEE Transactions on Signal Processing},
  volume  = {71},
  pages   = {1551--1562},
  year    = {2023},
  doi     = {10.1109/TSP.2023.3269911}
}

@article{fernando2010,
  author  = {T. Fernando and H. M. Trinh and L. Jennings},
  title   = {Functional observability and the design of minimum order linear functional observers},
  journal = {IEEE Transactions on Automatic Control},
  volume  = {55},
  number  = {5},
  pages   = {1268--1273},
  year    = {2010},
  doi     = {10.1109/TAC.2010.2042761}
}

@article{darouach2025,
  author  = {M. Darouach and T. Fernando},
  title   = {On functional observability and functional observer design},
  journal = {Automatica},
  volume  = {173},
  pages   = {112115},
  year    = {2025},
  doi     = {10.1016/j.automatica.2025.112115}
}

@article{scharf1994,
  author  = {Louis L. Scharf and Benjamin Friedlander},
  title   = {Matched subspace detectors},
  journal = {IEEE Transactions on Signal Processing},
  volume  = {42},
  number  = {8},
  pages   = {2146--2157},
  year    = {1994},
  doi     = {10.1109/78.301849}
}

@article{shannon1948,
  author  = {Claude E. Shannon},
  title   = {A mathematical theory of communication},
  journal = {Bell System Technical Journal},
  volume  = {27},
  number  = {3},
  pages   = {379--423},
  year    = {1948},
  doi     = {10.1002/j.1538-7305.1948.tb01338.x}
}

@book{cover2006,
  author    = {Thomas M. Cover and Joy A. Thomas},
  title     = {Elements of Information Theory},
  edition   = {2},
  publisher = {Wiley},
  address   = {Hoboken, NJ, USA},
  year      = {2006},
  doi       = {10.1002/047174882X}
}

@article{constantine2014,
  author  = {Paul G. Constantine and Eric Dow and Qiqi Wang},
  title   = {Active subspace methods in theory and practice: Applications to Kriging surfaces},
  journal = {SIAM Journal on Scientific Computing},
  volume  = {36},
  number  = {4},
  pages   = {A1500--A1524},
  year    = {2014},
  doi     = {10.1137/130916138}
}

@article{constantine2017,
  author  = {Paul G. Constantine and Paul Diaz},
  title   = {Global sensitivity metrics from active subspaces},
  journal = {Reliability Engineering \& System Safety},
  volume  = {162},
  pages   = {1--13},
  year    = {2017},
  doi     = {10.1016/j.ress.2017.01.013}
}

@book{kay1998fundamentals,
  author    = {Steven M. Kay},
  title     = {Fundamentals of Statistical Signal Processing, Volume II: Detection Theory},
  publisher = {Prentice Hall},
  address   = {Upper Saddle River, NJ, USA},
  year      = {1998}
}

@article{DingLiLiu2026ProjectionFaultDiagnosis,
  author  = {Ding, Steven X. and Li, Linlin and Liu, Tianyu},
  title   = {An Alternative Paradigm of Fault Diagnosis in Dynamic Systems:
             Orthogonal Projection-Based Methods},
  journal = {Automatica},
  volume  = {183},
  pages   = {112637},
  year    = {2026},
  doi     = {10.1016/j.automatica.2025.112637}
}

@book{Kato1995,
  author    = {Kato, Tosio},
  title     = {Perturbation Theory for Linear Operators},
  publisher = {Springer-Verlag},
  address   = {Berlin},
  year      = {1995}
}

@book{Vinnicombe2000,
  author    = {Vinnicombe, Glenn},
  title     = {Uncertainty and Feedback: H-Infinity Loop-Shaping and the Nu-Gap Metric},
  publisher = {World Scientific},
  year      = {2000}
}

@article{GeorgiouSmith1990,
  author  = {Georgiou, Tryphon T. and Smith, Malcolm C.},
  title   = {Optimal Robustness in the Gap Metric},
  journal = {IEEE Transactions on Automatic Control},
  volume  = {35},
  pages   = {673--686},
  year    = {1990}
}

@article{LiDing2020GapMetric,
  author  = {Li, Linlin and Ding, Steven X.},
  title   = {Gap Metric Techniques and Their Application to Fault Detection
             Performance Analysis and Fault Isolation Schemes},
  journal = {Automatica},
  volume  = {118},
  pages   = {109029},
  year    = {2020}
}
\end{document}